\documentclass[a4paper]{article}
\pdfoutput=1 % Make sure it compiles on arxiv

\usepackage[left=3cm, right=3cm, top=2.5cm, bottom=3cm]{geometry}

\usepackage{braket}
\usepackage{amsmath,amssymb,amsfonts}
\usepackage{amsthm}
\usepackage{mathtools}
\usepackage{graphicx}
\usepackage{xcolor}
\usepackage{comment}
\usepackage{dsfont}
\usepackage{enumitem}
\usepackage{authblk}
\usepackage{autobreak}
\usepackage[colorlinks=true,allcolors=blue]{hyperref}
\usepackage[capitalize,nameinlink]{cleveref}
\usepackage[style=phys,articletitle=true,biblabel=brackets,chaptertitle=false,pageranges=false]{biblatex}

\title{Macroscopicity and observational deficit\\in states, operations, and correlations\thanks{This work is dedicated to Professor Ryszard Horodecki on the occasion of his 80th birthday, in deep admiration and gratitude for his pioneering contributions to quantum information theory, his profound influence on the field, and his unwavering support and encouragement of young researchers.}}

\author{Teruaki Nagasawa\thanks{teruaki.nagasawa@nagoya-u.jp}}
\author{Eyuri Wakakuwa\thanks{e.wakakuwa@gmail.com}}
\author{Kohtaro Kato\thanks{kokato@i.nagoya-u.ac.jp}}
\author{Francesco Buscemi\thanks{buscemi@nagoya-u.jp}}

\affil{Department of Mathematical Informatics, Nagoya University,\\ Furo-cho Chikusa-ku, Nagoya 464-8601, Japan}

\newtheorem{theorem}{Theorem}[section]
\newtheorem{definition}[theorem]{Definition}
\newtheorem{lemma}[theorem]{Lemma}
\newtheorem{proposition}[theorem]{Proposition}
\newtheorem{corollary}[theorem]{Corollary}
\newtheorem{example}[theorem]{Example}

\newtheorem{remark}[theorem]{Remark}

\newcommand{\f}{\forall}

\newcommand{\ot}{\otimes}%テンソル積
\newcommand{\mr}{\mathrm}%数式内でローマン体
\newcommand{\mf}{\mathfrak}%数式内でドイツ文字
\newcommand{\mcl}{\mathcal}%数式内で筆記体
\newcommand{\mbb}{\mathbb}%数式内で黒板文字
\newcommand{\lam}{\lambda}%ギリシャ文字の小文字のラムダ
\newcommand{\gm}{\gamma}%ギリシャ文字のガンマ
\newcommand{\sn}{\operatorname{span}}

\def\sH{\mathcal{H}}

\def\openone{\mathds{1}}

\newcommand{\povm}[1]{\boldsymbol{#1}}

\renewcommand{\P}{{P}}

\renewcommand{\set}[1]{\mathcal{#1}}

\newcommand{\ketbra}[1]{|#1\rangle\!\!\; \langle #1 |}
\newcommand{\tr}[1]{\operatorname{Tr}\!\left[#1\right]}

\renewcommand{\leq}{\leqslant}
\renewcommand{\geq}{\geqslant}

\renewcommand{\ge}{\geqslant}

\makeatletter
\newcommand*\bigcdot{\mathpalette\bigcdot@{.5}}
\newcommand*\bigcdot@[2]{\mathbin{\vcenter{\hbox{\scalebox{#2}{$\m@th#1\bullet$}}}}}
\makeatother

\definecolor{KKgreen}{RGB}{0,200,0}

\begin{document}
\maketitle
\begin{abstract}

    To understand the emergence of macroscopic irreversibility from microscopic reversible dynamics, the idea of coarse-graining plays a fundamental role. In this work, we develop a unified inferential framework for \textit{macroscopic states}, that is, coarse descriptions of microscopic quantum systems that can be inferred from macroscopic measurements. Building on quantum statistical sufficiency and Bayesian retrodiction, we characterize macroscopic states through equivalent abstract (algebraic) and explicit (constructive) formulations.

    Central to our approach is the notion of \textit{observational deficit}, which quantifies the degree of irretrodictability of a state relative to a prior and a measurement. This leads to a general definition of macroscopic entropy as an inferentially grounded measure of asymmetry under Bayesian inversion. We formalize this structure in terms of \textit{inferential reference frames}, defined by the pair consisting of a prior and a measurement, which encapsulate the observer’s informational perspective.
    
    We then formulate a resource theory of microscopicity, treating macroscopic states as free states and introducing a hierarchy of \textit{microscopicity-non-generating operations}. This theory unifies and extends existing resource theories of coherence, athermality, and asymmetry.

    Finally, we apply the framework to study quantum correlations under observational constraints, introducing the notion of \textit{observational discord} and deriving necessary and sufficient conditions for their vanishing in terms of information recoverability.

    This work is dedicated to Professor Ryszard Horodecki on the occasion of his 80th birthday, in deep admiration and gratitude for his pioneering contributions to quantum information theory.

\end{abstract}

\section{Introduction}
Since von Neumann’s introduction of quantum entropy~\cite{von1955mathematical}, significant efforts have been made to extend classical thermodynamic principles to quantum systems. One drawback of the original von Neumann entropy is that it does not naturally connect to macroscopic phenomena, such as the second law of thermodynamics. To fill the gap between macroscopic phenomena and microscopic theory, von Neumann proposed another type of entropy called \textit{macroscopic entropy}~\cite{von1955mathematical,vonNeumann1929translation}. This concept, long forgotten, has recently been rediscovered and mathematically generalized into what is known as \textit{observational entropy}~\cite{safranek2019a,safranek2019b,safranek2021brief,buscemi2022observational}, finding applications spanning quantum information, quantum thermodynamics, quantum inference, and resource theories~\cite{bai2024observational,riera2020finite,safranek2020classical,deutsch2020probabilistic,faiez2020typical,nation2020snapshots,strasberg2021clausius,hamazaki2022speed,modak2022observational,sreeram2023witnessing,schindler2023continuity,safranek2023work,safranek2023measuring,safranek2023ergotropic,teixido2024entropic,nagasawa2024generic}. 

The key idea is that we ultimately always observe quantum systems through the lens of a macroscopic measurement, which von Neumann defines as the list of all physical quantities that can be simultaneously measured by the observer. The latter, from a mathematical point of view, is simply represented by a positive operator-valued measure (POVM), which corresponds to the finest possible measurement. However, such a finest measurement is still not sufficient to extract all information for all microscopic states at once: i.e., there are situations in which information about the microscopic state is irretrievably lost. This is the reason why observational entropy can increase under unitary time evolution, while von Neumann microscopic entropy remains invariant~\cite{vonNeumann1929translation,von1955mathematical,nagasawa2024generic}. In this sense, the use of observational entropy has been advocated to explain the emergence of irreversible processes from an underlying reversible theory, including, for example, the increase in entropy observed in closed systems~\cite{nagasawa2024generic}. Beyond its foundational implications for entropy and irreversibility, the framework of observational entropy may also offer a novel perspective on quantum correlations.

Closely related to the above discussion is the mathematical fact that observational entropy is never less than the von Neumann entropy. A state whose observational entropy is equal to its von Neumann entropy is called {\it macroscopic}~\cite{buscemi2022observational, nagasawa2024generic}. Operationally, macroscopic states are those that are fully known to the macroscopic observer, i.e., states that can be fully recovered from the macroscopically accessible data alone~\cite{bai2024observational}. From such a Bayesian statistical perspective, the macroscopic state corresponding to a given microscopic state is nothing but the \textit{Bayesian inverse} retrodicted from the measurement outcomes distribution~\cite{PARZYGNAT-russo-2022-noncomm-Bayes-theorem,Parzygnat2023axiomsretrodiction,bai2024quantum} by means of the Petz recovery map~\cite{Parzygnat2023axiomsretrodiction,bai2024quantum}. In this sense, the observational entropy precisely becomes a measure of ``how irretrodictable'' a state is from the outcomes of a given measurement~\cite{bai2024observational,nagasawa2024generic}. 

The aim of this paper is to develop a mathematical framework for macroscopic states by extending the notion of observational entropy to a more general, retrodictive setting that incorporates arbitrary POVMs and quantum priors. Central to this extension is the introduction of \textit{observational deficit}, a measure of a state's irretrodictability relative to a given prior. Macroscopic states are then defined as those with zero observational deficit. Equivalently, they are the fixed points of a \textit{coarse-graining map}, i.e., the composition of a quantum-to-classical channel representing a POVM measurement and its associated Petz recovery map, constructed with respect to the chosen prior. As shown in Ref.~\cite{bai2024quantum}, this composition yields the optimal Bayesian inference that can be drawn from the measurement outcomes, grounding the notion of macroscopicity in a principled inferential framework.

Leveraging the theory of quantum statistical sufficiency~\cite{petz1986sufficient, petz1988sufficiency,petz2003monotonicity,jencova-petz-2006-sufficiency-survey,jenvcova2024recoverability}, we derive several complete and equivalent characterizations of the set of macroscopic states. A central result is the following: for any POVM and any quantum prior, there exists a unique projection-valued measure (PVM), that is at the same time compatible with the given POVM and that can be measured without disturbing the prior (the rigorous definition is given below in Definition~\ref{definition:mppp}). We refer to this PVM as the \textit{maximal projective post-processing} of the POVM with respect to the prior. Its existence and uniqueness are established by introducing a partial order on POVMs under classical post-processing. This structure may be of independent interest, as PVMs derived in this way can be interpreted as analogous to reference frames, but for inference rather than symmetry: they encode the minimal yet sufficient information required to perform optimal retrodiction, and are used to quantify the inferential asymmetry, or ``irretrodictability''~\cite{watanabe65}, of a given process. For this reason, we refer to the pair comprising the maximal PVM and the prior as the \textit{inferential reference frame} corresponding to a specific macroscopic observer.

Next, we investigate macroscopic states from a resource-theoretic perspective, aiming to characterize macroscopic irreversibility, such as that captured by the second law of thermodynamics, as a form of statistical irreversibility arising from information loss in macroscopic measurements~\cite{watanabe55,watanabe65}. We define a \textit{resource theory of microscopicity}, by specifying the corresponding resource-destroying map~\cite{liu2017resource} and introducing the relative entropy of microscopicity as a quantitative measure of irretrodictability: i.e., the extent to which microscopic details are irrecoverable from macroscopic observations. Within this framework, we characterize a hierarchy of free operations, which we name \textit{macroscopic operations}. Moreover, we show that this resource theory unifies and generalizes several prominent resource theories, including those of coherence~\cite{aberg2006quantifying,baumgratz2014quantifying,streltsov2017colloquium}, athermality~\cite{brandao2013resource,brandao2015second,buscemi-2015-fully-quantum-second-laws,Gour:2018aa-jennings-buscemi-2018}, and asymmetry~\cite{bartlett2007reference,gour2008resource,gour2009measuring}.

Building on this foundation, we introduce and investigate \emph{observational discord}, a measure designed to quantify correlations from an observer's perspective. We then establish the necessary and sufficient conditions for this correlation to vanish, framing them in terms of information recoverability. A key insight from our findings is that quantum correlations, such as \emph{entanglement}~\cite{einstein1935can,horodecki2009quantum}, \emph{deficit}~\cite{oppenheim2002thermodynamical,horodecki2003local,devetak2005distillation,horodecki2005local}, and \emph{discord}~\cite{ollivier2001quantum,modi2012classical,ming2018quantum}, should not be viewed as absolute properties that are merely present or absent. Instead, we show they are context-dependent resources whose very visibility and utility are determined by the observer’s \emph{(quantum) reference frame}~\cite{bartlett2007reference,gour2008resource,fewster2025quantum}. This observer-dependent perspective opens new avenues for investigating quantum information processing in realistic scenarios, which are often characterized by limited control and imperfect access to quantum systems.

This paper is organized as follows. In Section~\ref{section:notations}, we introduce the basic notation and concepts used throughout the work. Section~\ref{section:macroscopic} presents the key definitions of observational deficit, coarse-graining maps, and macroscopic states. In Section~\ref{section:characterization}, we define projective post-processings and prove the existence and uniqueness of the maximal projective post-processing for any POVM and prior. We then analyze the algebraic structure of the set of macroscopic states with general quantum priors and characterize them in terms of these projective measurements. In Section~\ref{section:microscopicity}, we develop the resource theory of microscopicity, treating macroscopic states as free and introducing a hierarchy of macroscopic operations and resource measures quantifying retrodictive irreversibility. Finally, in Section~\ref{section:correlation}, we introduce the concept of observational discord and derive the necessary and sufficient conditions under which it vanishes.

\section{Notations}
\label{section:notations}
In this paper, we consider a quantum system with a finite $d$-dimensional Hilbert space $\sH$. The set of linear operators on a Hilbert space $\sH$ is denoted by $\mcl{L}(\sH)$. We also denote the set of quantum states on $\mcl{H}$ by $\mcl{S}(\sH)$. The maximally mixed (or \textit{uniform}) state is denoted by $u\coloneq \openone/d$. Furthermore, for a \textit{completely positive trace-preserving} (CPTP) map (or \textit{channel}) $\mcl{E}$, we write as $\mcl{E}^*$ its adjoint with respect to the Hilbert-Schmidt inner product. Any observation on the quantum system can be represented by a \textit{positive operator-valued measure} (POVM) $\povm{P}=\{P_x\}_{x\in\set{X}}$, i.e., a family of positive semidefinite operators $P_x\geq 0$ such that $\sum_xP_x=\openone$. When all elements of a POVM are projections, i.e., $P_xP_{x'}=\delta_{x,x'}P_x$, the former is called a \textit{projection-valued measure} (PVM). In addition, POVMs can be put in a one-to-one correspondence with quantum-classical channels of the form
\begin{align}\label{eq:qc-channel}
    \mcl{M}_{\povm{P}}(\bigcdot)\coloneq \sum_{x\in\set{X}}\tr{P_x\;\bigcdot}\ketbra{x}\;,
\end{align}
where $\ket{x}$ are a set of orthonormal unit vectors in a suitable Hilbert space representing the classical outcomes.
The \emph{Umegaki quantum relative entropy}~\cite{umegaki-q-rel-ent-1961,umegaki1962conditional} is defined as
\begin{align}
    D(\rho\|\sigma)\coloneq \tr{\rho(\log\rho-\log\sigma)}\;,
\end{align}
where $\sigma$ (which we assume invertible for simplicity) is a reference (or \textit{prior}) operator.
Then, the von Neumann (microscopic) entropy is defined by
\begin{align}
    S(\rho)&\coloneq -D(\rho\|\openone)=-\tr{\rho\log\rho}\;.
\end{align}
The logarithm is taken in base 2.

Let us now consider an arbitrary quantum channel $\mcl{E}$ and a state $\gm$, which in the following is either assumed to be invertible, otherwise the whole discussion can be restricted to its support. Then, the corresponding \emph{Petz map} is defined as
\begin{align}
    \mcl{R}_{\mcl{E}, \gm}(\bigcdot) \coloneq  \gm^{\frac{1}{2}}\mcl{E}^{*} \left[\mcl{E}(\gm)^{-\frac{1}{2}} (\bigcdot) \mcl{E}(\gm)^{-\frac{1}{2}}\right]\gm^{\frac{1}{2}}\;.
\end{align}
A central result in quantum mathematical statistics is the following statement about sufficient statistics, which was characterized by Petz already in the 1980s~\cite{petz1986sufficient,petz1988sufficiency,petz2003monotonicity,petz2007quantum}: for any quantum state $\rho$, the condition
\begin{align}
    D(\rho\|\gm)=D(\mcl{E}(\rho)\|\mcl{E}(\gm))\label{equation:eq-dpi}
\end{align}
is equivalent to
\begin{align}
    \mcl{R}_{\mcl{E},\gm}\circ\mcl{E}(\rho)=\rho\;.\label{equation:eq-petz}
\end{align}
The above theorem gives the equality condition for the universally valid \emph{quantum data processing inequality} (DPI):
\begin{align}
    D(\rho\|\gm)\geq D(\mcl{E}(\rho)\|\mcl{E}(\gm))\;.
\end{align}

Given two POVMs $\povm{P}=\{P_x\}_{x\in\set{X}}$ and $\povm{Q}=\{Q_y\}_{y\in\set{Y}}$, we write 
\begin{align}\label{eq:post-proc-preorder}
    \povm{Q}\preceq\povm{P}
\end{align}
whenever there exists a conditional probability distribution $p(y|x)$ such that 
\begin{align}
    Q_y=\sum_xp(y|x)P_x
\end{align}
for all $y\in\set{Y}$. In this case, we say that $\povm{Q}$ is a classical post-processing of $\povm{P}$~\cite{martens1990nonideal,buscemi-2005-clean-POVMs}. We will often refer to the following lemma.

\begin{lemma}[\cite{nagasawa2024generic}]
\label{lemma:groupings}
    Suppose that $\povm{Q}=\{Q_y\}_{y\in\set{Y}}$ is a PVM, i.e., $Q_y Q_{y'} = \delta_{yy'} Q_y$ for all $y, y' \in \set{Y}$. Suppose also that there exists another POVM $\povm{P}=\{P_x\}_{x\in\set{X}}$ such that $\povm{Q}\preceq \povm{P}$. Then, the post-processing transforming $\povm{P}$ into $\povm{Q}$ is deterministic, i.e., 
    \begin{align}
        p(y|x)\in\{0,1\}
    \end{align}
    for all $x$ and $y$.
\end{lemma}

\section{Macroscopic states}
\label{section:macroscopic}
Inspired by the terminology introduced in Ref.~\cite{horodecki2005local}, we begin with the following definition~\cite{bai2024observational}.
\begin{definition}[Observational deficit]
\label{definition:deficit}
    For any POVM $\povm{P}=\{P_x\}_{x\in\set{X}}$ and any pair of quantum states $\rho$ and $\gm>0$, the \emph{observational deficit} of $\povm{P}$ with respect to the pair (dichotomy) $(\rho,\gamma)$ is defined as
    \begin{align}
        \delta_{\povm{P}}(\rho\|\gm)\coloneq D(\rho\|\gm)-D(\mcl{M}_{\povm{P}}(\rho)\|\mcl{M}_{\povm{P}}(\gm))\geq0\;.
    \end{align}
\end{definition}
The observational deficit defined here does not merely quantify information loss in a statistical sense: it also implicitly reflects the degradation of observable correlations under macroscopic measurements. This view, which directly connects with the notion of \textit{quantum deficit} developed by Horodecki and collaborators~\cite{oppenheim2002thermodynamical,horodecki2003local,devetak2005distillation,horodecki2005local}, will be explored in Section~\ref{section:correlation}.

The observational deficit is zero if and only if the quantum-classical channel $\mcl{M}_{\povm{P}}$ is sufficient with respect to the quantum dichotomy $(\rho, \gm)$, i.e., the corresponding Petz map $\mcl{R}_{\mcl{M}_{\povm{P}}, \gm}$ satisfies both $[\mcl{R}_{\mcl{M}_{\povm{P}}, \gm}\circ\mcl{M}_{\povm{P}}](\rho)=\rho$ and $[\mcl{R}_{\mcl{M}_{\povm{P}}, \gm}\circ\mcl{M}_{\povm{P}}](\gamma)=\gamma$, the latter by construction~\cite{petz1986sufficient, petz1988sufficiency,petz2003monotonicity,jencova-petz-2006-sufficiency-survey,jenvcova2024recoverability}.

Since the channel to be recovered, i.e., $\mcl{M}_{\povm{P}}$, is a quantum-to-classical measurement channel, its corresponding Petz transpose map is the unique solution to the \textit{minimum change principle}~\cite{bai2024quantum}. It thus has a compelling operational interpretation as the quantum analog of Bayes' rule even when $\delta_{\povm{P}}(\rho\|\gm)>0$. This observation motivates us to define\footnote{Note that the definition we give here of ``coarse-graining'' differs from that used in~\cite{buscemi2022observational}.} \emph{coarse-graining maps} as follows.

\begin{definition}[Coarse-graining maps]
\label{definition:coarse-graining-map}
    Let $\povm{P}$ be a POVM and let $\gm$ be a quantum state. Then, the \emph{coarse-graining map} with respect to $\povm{P}$ and $\gm$ is defined as
    \begin{align}
        \mcl{C}_{\povm{P},\gm}(\bigcdot)&\coloneq [\mcl{R}_{\mcl{M}_{\povm{P}}, \gm}\circ\mcl{M}_{\povm{P}}](\bigcdot)\\
        &=\sum_x\tr{P_x(\bigcdot)}\frac{\gm^{\frac{1}{2}}P_x\gm^{\frac{1}{2}}}{\tr{P_x\gm}}\;.\label{eq:CGMdfn}
    \end{align}
\end{definition}

It is immediate to verify that all coarse-grainings are channels of the measure-and-prepare kind and, as such, destroy entanglement when applied locally~\cite{horodecki2003entanglement}. This observation will play an important role in Section~\ref{section:correlation}.

\begin{definition}[Coarse-grained state and macroscopic state]
\label{definition:coarse-grained}
    Let $\rho$ be a quantum state. Then, $\mcl{C}_{\povm{P},\gm}(\rho)$ is called the \emph{coarse-grained state} corresponding to $\rho$. Furthermore, when 
    \begin{align}
        \mcl{C}_{\povm{P},\gm}(\rho)=\rho\;, 
    \end{align}
    $\rho$ is said to be \emph{macroscopic} with respect to the POVM $\povm{P}$ and the quantum prior $\gm$.
\end{definition}

As explained below, given an observation $\povm{P}$ and a prior $\gamma$, $\mcl{C}_{\povm{P},\gm}$ corresponds to the coarse-graining channel acting on quantum states, while $\mcl{C}_{\povm{P},\gm}^*$ corresponds to the coarse-graining \textit{post-processing} acting on POVMs outcomes.

\begin{remark}[Coarse-graining and POVM post-processing]
    Let $\povm{Q}=\{Q_y\}_y$ be a POVM. Then, 
    \begin{align}
        \mcl{C}_{\povm{P},\gm}^*(Q_y)&=\sum_x\tr{Q_y\frac{\gm^{\frac{1}{2}}P_x\gm^{\frac{1}{2}}}{\tr{P_x\gm}}}P_x\\
        &=\sum_xq(y|x)P_x\;,
    \end{align}
    where $q(y|x)\coloneq \tr{Q_y\left(\gm^{\frac{1}{2}}P_x\gm^{\frac{1}{2}}/\tr{P_x\gm}\right)}$. By definition, this means that for any $\povm{Q}$
    \begin{align}
        \mcl{C}_{\povm{P},\gm}^*(\povm{Q})\preceq\povm{P}\;.
    \end{align}
    Thus, the adjoint of the coarse-graining map for quantum states corresponds to POVMs post-processing. The interpretation is very natural: if we perform a further observation (the POVM $\povm{Q}$) on a quantum state already coarse-grained under $\povm{P}$, it is as if we were observing a classical post-processing of $\povm{P}$ on the original state (before the coarse-graining). That is, after the state has been coarse-grained with respect to a POVM $\povm{P}$, it only contains information that can be perfectly inferred from $\povm{P}$, and nothing more.
\end{remark} 

\subsection{Observational entropy}
\label{subsection:observational}
\begin{definition}[Observational entropy]
\label{remark:OE}
    Let $\povm{P}=\{P_x\}_{x\in\set{X}}$ be a POVM. Then for arbitrary $\rho\in\mcl{S}(\sH)$,
    \begin{align}
        S_{\povm{P}}(\rho)&\coloneq S(\rho)+\delta_{\povm{P}}(\rho\|u)
    \end{align}
    is called the \emph{observational entropy} of $\rho$ with respect to $\povm{P}$. (Notice the uniform prior.)
\end{definition}

Notice that while von Neumann originally defined macroscopic entropy only for PVMs~\cite{von1955mathematical,vonNeumann1929translation}, the \emph{observational entropy} extends the definition for arbitrary POVMs~\cite{safranek2019a,safranek2019b,safranek2021brief,buscemi2022observational}.

In the case of observational entropy, as conventionally defined, the prior state is assumed to be the uniform state $u$, but more generally, we can consider arbitrary prior distributions~\cite{bai2024observational,schindler2025unification}. For example, the prior distribution $\gm$ can be viewed as a generalization of the equilibrium state.

Von Neumann entropy is an entropy derived from thermo-statistical discussions based on the assumption that the second law of thermodynamics is universally valid~\cite{von1955mathematical}. On the other hand, macroscopic entropy is an entropy calculated under the additional assumption that states that cannot be distinguished by a macroscopic observer (PVM) are considered to be the same state.

In order to formalize this, let us assume that a macroscopic observer corresponding to a PVM $\povm{\Pi}=\{\Pi_y\}_{y=1}^m$ cannot distinguish between $\rho$ and $\rho'$. This means that all PVM elements, $\Pi_1, \Pi_2,\dots, \Pi_m$, have the same expectation value:
\begin{align}
\label{equation:macro}
    \tr{\rho\ \Pi_y}=\tr{\rho'\ \Pi_y}\;,\qquad\f y\;.
\end{align}
Now we have, for any $\rho\in\mcl{S}(\sH)$,
\begin{align}
\label{equation:macro-state}
    \tr{\rho\ \Pi_y}=\tr{\left(\sum_{n=1}^m \tr{\rho \Pi_n}\frac{\Pi_n}{\tr{\Pi_n}}\right)\Pi_y}\;,\qquad\f y\;.
\end{align}
Thus, $\rho$ and $\rho_{\povm{\Pi},u}\coloneq\sum_{n=1}^m \tr{\rho \Pi_n}\Pi_n/\tr{\Pi_n}$ are indistinguishable for macroscopic observer $\povm{\Pi}$ and have the same macroscopic entropy. Here, the macroscopic entropy of the quantum state $\rho$ is given by the von Neumann entropy of the corresponding macroscopic state $\rho_{\povm{\Pi},u}$~\cite{von1955mathematical}:
\begin{align}
    S(\rho_{\povm{\Pi},u})&=-\tr{\rho_{\povm{\Pi},u}\log\rho_{\povm{\Pi},u}}\\
    &=-\sum_{n=1}^m\tr{\Pi_n\rho}\log \frac{\tr{\Pi_n\rho}}{\tr{\Pi_n}}\\
    &=S_{\povm{\Pi}}(\rho)\\
    &=S_{\povm{\Pi}}(\rho_{\povm{\Pi},u})\;.
\end{align}
If we consider such a macroscopic state $\rho_{\povm{\Pi},u}$ as the initial state, we can then show that observational entropy will generally increase with unitary time-evolution~\cite{strasberg-winter-2021-PRX-quantum,nagasawa2024generic}. On the contrary, in the microscopic case, von Neumann entropy cannot increase (in fact, it remains constant) under unitary time evolution.

\section{Equivalent characterizations of macroscopic states and inferential reference frames}
\label{section:characterization}

To systematically understand which features of a quantum state, including correlations, survive macroscopic coarse-graining, we now establish equivalent characterizations of macroscopicity. These will allow us to later identify the correlation content of macroscopic versus microscopic states. 

The theorem proved in this section, which generalizes the known case of rank-1 PVM~\cite{ohya1993quantum,hayashi2017quantum}, is one of the main results of this paper. It provides four, from purely algebraic to explicitly constructive, equivalent characterizations of macroscopic states for arbitrary POVMs and priors. Crucial for our discussion is the concept of \textit{maximal projective post-processing} (MPPP), which plays a central role in our framework as an ``inferential symmetry reference frame'', i.e., a projection-valued measure that captures the information recoverable under both a measurement and a prior. The following definition uses the post-processing preorder defined in Eq.~\eqref{eq:post-proc-preorder}.

\begin{definition}[Maximal projective post-processing (MPPP)]
\label{definition:mppp}
    A PVM $\povm{\Pi}$ is called a \emph{$\gamma$-commuting projective post-processing} (PPP) of $\povm{P}$ if $\povm{\Pi}\preceq\povm{P}$ and all its elements commute with $\gamma$.
    A $\gamma$-commuting PPP $\povm{\Pi}$ of $\povm{P}$ is said to be \emph{maximal} if $\povm{\Pi}'\preceq\povm{\Pi}$ for any $\gamma$-commuting PPP $\povm{\Pi}'$ of $\povm{P}$.
\end{definition}

The proof that a $\gamma$-commuting MPPP always exists unique for any pair $(\povm{P},\gamma)$ is postponed to Section~\ref{subsection:maximal}. In what follows, we assume that the prior $\gamma$ is invertible and that the POVM $\povm{P}$ has no zero-elements, which implies that also $\mcl{M}_{\povm{P}}(\gm)$ is invertible. We are now ready to state the theorem.

\begin{theorem}[Macroscopic states with general quantum prior]
\label{theorem:macroscopic}
    Let $\povm{P}=\{P_x\}_{x\in\set{X}}$ be a POVM, $\gm$ be an invertible quantum state, and $\povm{\Pi}_{\povm{P},\gm}\equiv\povm{\Pi}=\{\Pi_y\}_{y\in\set{Y}}$ be the corresponding MPPP, see Definition~\ref{definition:mppp}. Then, for any quantum state $\rho$, the following conditions are equivalent:
    \begin{enumerate}
        \item $\delta_{\povm{P}}(\rho\|\gm)=0$;\label{theorem:macroscopic:deficit}
        \item $\mcl{C}_{\povm{P},\gm}(\rho)=\rho$;\label{theorem:macroscopic:coarse-graining}
        \item $\Delta_{\povm{P},\gm}(\rho)=\rho$, where $\Delta_{\povm{P},\gm}$ is an idempotent CPTP map defined by 
        \begin{align}
            \Delta_{\povm{P},\gm}(\bigcdot)\coloneq\sum_{y\in\set{Y}}\tr{\Pi_y(\bigcdot)}\frac{\Pi_y\gm}{\tr{\Pi_y\gm}}\;;\label{eq:RDMrtm}
        \end{align}\label{theorem:macroscopic:resource-destroying}
        \item there exist coefficients $c_y\geq0$ such that $\rho=\sum_{y\in\set{Y}}c_y\Pi_y\gm$.\label{theorem:macroscopic:explicit}
    \end{enumerate}
\end{theorem}

By the equality condition of DPI, Eqs.~(\ref{equation:eq-dpi}) and~(\ref{equation:eq-petz}), we have~(\ref{theorem:macroscopic:deficit})$\iff$(\ref{theorem:macroscopic:coarse-graining}). We have~(\ref{theorem:macroscopic:resource-destroying})$\iff$(\ref{theorem:macroscopic:explicit}) by the definition of $\Delta_{\povm{P},\gm}$ and the fact that $\povm{\Pi}$ is a PVM whose elements commute with $\gm$. The proof of the remaining equivalence, i.e.,~(\ref{theorem:macroscopic:coarse-graining})$\iff$(\ref{theorem:macroscopic:resource-destroying}) is given below,  as Theorem~\ref{theorem:rdm-coarse}. 

\begin{remark}
    In \cite[Lemma 4.1]{petz2003monotonicity} it is stated that if 
    \begin{align}
        \delta_{\povm{P}}(\rho\|\gm)=0\;,\label{equation:deficit}
    \end{align}
    then
    \begin{align}
        [P_x,\gm]=0\label{equation:commute}
    \end{align}
    for all $x\in\set{X}$. However, the above does not hold in general. Indeed, regardless of how we choose the POVM $\povm{P}=\{P_x\}_{x\in\set{X}}$ and the prior $\gm$, the latter is \emph{always} a macroscopic state (by construction), even if we choose it so that $[P_x,\gm]\neq 0$. Only in the particular case when $\gamma=u$, then, all macroscopic states must commute with the initial POVM~\cite{nagasawa2024generic}.
\end{remark}

\begin{remark}
    In hindsight, knowing that $\delta_{\povm{P}_\gm}(\rho\|\gm)=0$ implies $[\rho,\gm]=0$, it is possible to characterize macroscopic states using, instead of $\povm{P}$, its pinched version $\povm{P}_\gm$---pinched with respect to an eigenbasis of $\gm$. Our approach has the advantage of showing, in an explicitly constructive way, that the condition $\delta_{\povm{P}_\gm}(\rho\|\gm)=0$ implies $[\rho,\gm]=0$.
\end{remark}

\subsection{Maximal projective post-processing}
\label{subsection:maximal}

In what follows, we prove the existence and uniqueness of the $\gamma$-commuting MPPP defined in Definition~\ref{definition:mppp}. The statement appears as Corollary~\ref{corollary:mppp-existence-uniqueness}, and its proof is broken into a few preceding Lemmas and a Proposition.

We fix an arbitrary prior $\gamma$ and an arbitrary POVM $\povm{P}=\{P_x\}_{x\in\mathcal{X}}$, and define $\mathcal{X}_+\coloneq \{x\in\mathcal{X}|P_x\neq0\}$. 
We begin with the following lemma:

\begin{lemma}\label{lemma:partition}
    Let $\povm{\Pi}=\{\Pi_y\}_{y\in\mathcal{Y}}$ be a set of projectors.
    Then, $\povm{\Pi}$ is a PVM satisfying $\povm{\Pi}\preceq \povm{P}$ if and only if there exists a disjoint partition of $\mathcal{X}_+$ into $\{\mathcal{X}_y\}_{y\in\mathcal{Y}}$ such that $\Pi_y=\sum_{x\in\mathcal{X}_y}P_x$ for all $y\in\mathcal{Y}$.
\end{lemma}

\begin{proof}
    The ``if'' part follows because
    \begin{align}
        \Pi_y&=\sum_{x\in\mathcal{X}_y}P_x\\
        &=\sum_{x\in\mathcal{X}}p(y|x)P_x\;,
    \end{align}
    where
    \begin{align}
        p(y|x)=
        \begin{cases}
            1 &(x\in\mathcal{X}_y)
            \\
            0 &(x\notin\mathcal{X}_y)
        \end{cases}
    \end{align}
    and
    \begin{align}
        \sum_{y\in\mathcal{Y}}\Pi_y&=\sum_{y\in\mathcal{Y}}\sum_{x\in\mathcal{X}_y}P_x\\
        &=\sum_{x\in\mathcal{X}}P_x\\
        &=\openone.
    \end{align}
     The ``only if'' part follows from Lemma \ref{lemma:groupings}.
     Indeed, if $\povm{\Pi}$ is a PVM satisfying $\povm{\Pi}\preceq \povm{P}$, which clearly has a unit eigenvector with eigenvalue 1, there exists a conditional probability distribution $\{p(y|x)\}$ such that $\Pi_y=\sum_{x\in\mathcal{X}}p(y|x)P_x$ and $p(y|x)\in\{0,1\}$ for all $x$ and $y$.
    Letting $\mathcal{X}_y\coloneq \{x\in\mathcal{X}|p(y|x)=1\}$, we complete the proof.
\end{proof}

\

\begin{definition}
    A partition $\{\mathcal{X}_y\}_{y\in\mathcal{Y}}$ of $\mathcal{X}_+$ is said to be \emph{$\gamma$-disconnected} if, for any $y\neq y'$, any $x\in\mathcal{X}_y$, and any $x'\in\mathcal{X}_{y'}$, it holds that $P_xP_{x'}=P_x\gamma P_{x'}=0$. It is said to be \emph{$\gamma$-connected} otherwise.
\end{definition}

\begin{lemma}\label{lmm:PPP}
    Given a partition $\{\mathcal{X}_y\}_{y\in\mathcal{Y}}$ of $\mathcal{X}_+$, let $\povm{Q}=\{Q_y\}_{y\in\mathcal{Y}}$ be a POVM such that $Q_y=\sum_{x\in\mathcal{X}_y}P_x$. Then, $\povm{Q}$ is a $\gamma$-commuting PPP of $\povm{P}$ if and only if the partition $\{\mathcal{X}_y\}_{y\in\mathcal{Y}}$ is $\gamma$-disconnected.
\end{lemma}

\begin{proof}
     To prove the ``if'' part, let $\{\mathcal{X}_y\}_{y\in\mathcal{Y}}$ be a $\gamma$-disconnected partition of $\mathcal{X}_+$ and define  $Q_y=\sum_{x\in\mathcal{X}_y}P_x$. For $y\neq y'$, we have
     \begin{align}   Q_yQ_{y'}&=\sum_{x\in\mathcal{X}_y}\sum_{x'\in\mathcal{X}_{y'}}P_xP_{x'}=0,\label{eq:QyQy}\\     Q_y\gamma Q_{y'}&=\sum_{x\in\mathcal{X}_y}\sum_{x'\in\mathcal{X}_{y'}}P_x\gamma P_{x'}=0.\label{eq:QygQy}
     \end{align}
     Noting that $\sum_{y\in\mathcal{Y}}Q_y=I$, we have, from (\ref{eq:QyQy}),
     \begin{align}
         {\rm Tr}[Q_y]=\sum_{y'\in\mathcal{Y}}{\rm Tr}[Q_yQ_{y'}]={\rm Tr}[Q_y^2].
     \end{align}
     Since $0\leq Q_y\leq I$, it follows that $Q_y$ are projectors, which are orthogonal due to (\ref{eq:QyQy}).
     Eq.~(\ref{eq:QygQy}) implies that $\gamma$ has no off-diagonal term with respect to $\{Q_y\}_{y\in\mathcal{Y}}$, which implies $[Q_y,\gamma]=0$.
     
     To prove the ``only if'' part, let $\{\mathcal{X}_y\}_{y\in\mathcal{Y}}$ be a given partition of $\mathcal{X}_+$, let $Q_y=\sum_{x\in\mathcal{X}_y}P_x$, and suppose that $\povm{Q}=\{Q_y\}_{y\in\mathcal{Y}}$ is a $\gamma$-commuting PPP of $\povm{P}$. Due to the orthogonality of the projectors, for $y\neq y'$, we have
     \begin{align}
         0={\rm Tr}[Q_yQ_{y'}]=\sum_{x\in\mathcal{X}_{y}}\sum_{x'\in\mathcal{X}_{y'}}{\rm Tr}[P_xP_{x'}].
     \end{align}
     Noting that $P_x\geq0$, this implies ${\rm Tr}[P_xP_{x'}]=0$ and thus $P_xP_{x'}=0$ for any $x\in\mathcal{X}_{y}$ and $x'\in\mathcal{X}_{y'}$. Additionally, the commutativity of $Q_y$ with $\gamma$ implies
     \begin{align}
         P_x\gamma P_{x'}=P_x\left(\sum_{y\in\mathcal{Y}}Q_y\gamma Q_y\right)P_{x'}=\sum_{y\in\mathcal{Y}}(P_xQ_y)\gamma (Q_yP_{x'}).
     \end{align}
     Since $P_xQ_y=0$ for any $x\notin\mathcal{X}_y$ due to the orthogonality of $\{Q_y\}_{y\in\mathcal{Y}}$, the above is equal to zero for any pair of $x$ and $x'$ that belong to different subsets $\mathcal{X}_{y}$ and $\mathcal{X}_{y'}$. This implies that $\{\mathcal{X}_y\}_{y\in\mathcal{Y}}$ is $\gamma$-disconnected.
\end{proof}

\

\begin{definition}\label{dfn:irreducibledisconnected}
    A $\gamma$-disconnected partition $\{\mathcal{X}_{y}\}_{y\in\mathcal{Y}}$ of $\mathcal{X}_+$ is said to be \emph{irreducible} if any of its strict refinements (in the sense of partition) is $\gamma$-connected. It is said to be \emph{reducible} otherwise.
\end{definition}

\begin{lemma}\label{lmm:finer}
    There exists an irreducible $\gamma$-disconnected partition of $\mathcal{X}_+$ that is finer than any other $\gamma$-disconnected partition of $\mathcal{X}_+$.
\end{lemma}

\begin{proof}
    Let $\{\mathcal{X}_{y}\}_{y\in\mathcal{Y}}$ be any $\gamma$-disconnected partition of $\mathcal{X}_+$ and let $\{\mathcal{X}_z^*\}_{z\in\mathcal{Z}}$ be any irreducible one. It suffices to prove that, for any $z\in\mathcal{Z}$, there exists $y\in\mathcal{Y}$ such that $\mathcal{X}_z^*\subseteq\mathcal{X}_{y}$. We prove this by contradiction. Suppose that there exist $z\in\mathcal{Z}$ and $y\in\mathcal{Y}$ such that $\mathcal{X}_{z,0}^*\coloneq \mathcal{X}_z^*\cap\mathcal{X}_y\neq\emptyset$ and $\mathcal{X}_{z,1}^*\coloneq \mathcal{X}_z^*\cap(\bigcup_{y'\neq y}\mathcal{X}_{y'})\neq\emptyset$. Since $\mathcal{X}_y$ and $\bigcup_{y'\neq y}\mathcal{X}_{y'}$ are $\gamma$-disconnected by assumption, $\mathcal{X}_{z,0}^*$ and $\mathcal{X}_{z,1}^*$ are $\gamma$-disconnected. Thus, $\{\mathcal{X}_{z,0}^*,\mathcal{X}_{z,1}^*\}\cup\{\mathcal{X}_z^*\}_{z'\in\mathcal{Z}\backslash\{z\}}$ is a $\gamma$-disconnected partition of $\mathcal{X}_+$, which contradicts the assumption that $\{\mathcal{X}_z^*\}_{z\in\mathcal{Z}}$ is irreducible.
\end{proof}

\

\begin{lemma}
The irreducible $\gamma$-disconnected partition of $\mathcal{X}_+$ is unique up to relabeling of the subsets.
\end{lemma}

\begin{proof}
Let $\{\mathcal{X}^*_z\}_{z\in\mathcal{Z}}$ and $\{\mathcal{X}^*_\theta\}_{\theta\in\Theta}$ be irreducible $\gamma$-disconnected partitions of $\mathcal{X}_+$. Due to Lemma \ref{lmm:finer}, for any $z\in\mathcal{Z}$, there exist $\theta\in\Theta$ and $z'\in\mathcal{Z}$ such that $\mathcal{X}^*_z\subseteq\mathcal{X}^*_\theta\subseteq\mathcal{X}^*_{z'}$. This relation holds only if $z=z'$ and $\mathcal{X}^*_z=\mathcal{X}^*_\theta$, thus $\{\mathcal{X}^*_z\}_{z\in\mathcal{Z}}$ and $\{\mathcal{X}^*_\theta\}_{\theta\in\Theta}$ are equal up to relabeling.
\end{proof}

\

\begin{proposition}
\label{proposition:mppp-existence}
    Let $\{\mathcal{X}_z\}_{z\in\mathcal{Z}}$ be a $\gamma$-disconnected partition of $\mathcal{X}_+$, and let $\povm{\Pi}=\{\Pi_z\}_{z\in\mathcal{Z}}$ be a POVM such that $\Pi_z=\sum_{x\in\mathcal{X}_z}P_x$ for each $z\in\mathcal{Z}$. Due to Lemma \ref{lmm:PPP}, $\povm{\Pi}$ is a $\gamma$-commuting PPP of $\povm{P}$ . Then, $\povm{\Pi}$ is maximal if and only if $\{\mathcal{X}_z\}_{z\in\mathcal{Z}}$ is irreducible. 
\end{proposition}

\begin{proof}
   To prove the ``if'' part, suppose that $\{\mathcal{X}_z\}_{z\in\mathcal{Z}}$ is irreducible. Let $\povm{\Pi}'=\{\Pi_y'\}_{y\in\mathcal{Y}}$ be any $\gamma$-commuting PPP of $\povm{P}$. Due to Lemma \ref{lemma:partition} and Lemma \ref{lmm:PPP}, there exists a $\gamma$-disconnected partition $\{\mathcal{X}_y'\}_{y\in\mathcal{Y}}$ of $\mathcal{X}_+$ such that $\Pi_y'=\sum_{x\in\mathcal{X}_y'}P_x$. Due to Lemma \ref{lmm:finer}, $\{\mathcal{X}_z\}_{z\in\mathcal{Z}}$ is finer than $\{\mathcal{X}_y'\}_{y\in\mathcal{Y}}$, which implies that there exists a partition $\{\mathcal{Z}_y\}_{y\in\mathcal{Y}}$ of $\mathcal{Z}$ such that $\mathcal{X}_y'=\cup_{z\in\mathcal{Z}_y}\mathcal{X}_z$. Hence, we have $\Pi_y'=\sum_{z\in\mathcal{Z}_y}\Pi_z$, which implies $\povm{\Pi}'\preceq\povm{\Pi}$.
   
   To prove the ``only if'' part, suppose, on the contrary, that $\{\mathcal{X}_z\}_{z\in\mathcal{Z}}$ is reducible. Let $\{\mathcal{X}_\theta^*\}_{\theta\in\Theta}$ be the irreducible $\gamma$-disconnected partition of $\mathcal{X}_+$. Reducibility of $\{\mathcal{X}_z\}_{z\in\mathcal{Z}}$ implies that there exist $\theta_1\in\Theta$ and $z_1\in\mathcal{Z}$ such that $\mathcal{X}_{\theta_1}^*\subsetneq\mathcal{X}_{z_1}$.
   Let $\povm{\Pi}^*=\{\Pi_\theta^*\}_{\theta\in\Theta}$ be a $\gamma$-commuting PPP of $\povm{P}$ such that $\Pi_\theta^*=\sum_{x\in\mathcal{X}_\theta^*}P_x$, which is maximal due to the first half of this proof. It suffices to prove that no classical post-processing transforms $\Pi$ to $\Pi^*$. Indeed, if such a classical post-processing exists, due to Lemma \ref{lemma:partition}, it is represented by a partition $\{\mathcal{Z}_{\theta}\}_{\theta\in\Theta}$ of $\mathcal{Z}$ such that $\Pi_\theta^*=\sum_{z\in\mathcal{Z}_\theta}\Pi_z$ for any $\theta\in\Theta$. This implies that $\mathcal{X}_\theta^*=\cup_{z\in\mathcal{Z}_\theta}\mathcal{X}_z$, so there exists $z_2\in\mathcal{Z}$ that satisfies $\mathcal{X}_{z_2}\subseteq\mathcal{X}_{\theta_1}^*$. Hence, we have $\mathcal{X}_{z_2}\subsetneq\mathcal{X}_{z_1}$, which is a contradiction.
\end{proof}

\

\begin{corollary}
\label{corollary:mppp-existence-uniqueness}
For any POVM $\povm{P}$, the maximal $\gamma$-commuting projective post-processing $\povm{\Pi}_{\povm{P},\gm}$ of $\povm{P}$ exists and is unique up to invertible relabeling of the measurement outcomes. 
\end{corollary}

\subsection{Fixed points of the coarse-graining map}
\label{subsection:fixed}
In this subsection, we focus on the algebraic structure of coarse-graining maps and their adjoints. We prove that the fixed-point set of the adjoint map of a coarse-graining is completely characterized by the MPPP:

\begin{theorem}
\label{theorem:dual}
Let $\povm{\Pi}_{\povm{P},\gm}=\{\Pi_y\}_{y\in\set{Y}}$ be the MPPP with respect to POVM $\povm{P}=\{P_x\}_{x\in\set{X}}$ and prior $\gm$. Moreover, let   
\begin{align}
    \mf{F}_{\mcl{C}_{\povm{P},\gm}^*}\coloneq\left\{A\in\mcl{L}(\sH)\mid\mcl{C}_{\povm{P},\gm}^*(A)=A\right\}
\end{align}
be the set of fixed points of $\mcl{C}_{\povm{P},\gm}^*$. Then, we have
    \begin{align}
        \mf{F}_{\mcl{C}_{\povm{P},\gm}^*}=\sn\left\{\Pi_y\right\}_{y\in\set{Y}}\;.
    \end{align}
\end{theorem}

(The proof of the above theorem will be given after the proof of Lemma~\ref{lemma:dual} below.) Therefore, in particular, we can interpret $\povm{\Pi}_{\povm{P},\gamma}$ as a sort of \emph{inferential reference frame} for the macroscopic observer represented by the pair $(\povm{P},\gamma)$: specifically, the fixed-point algebra of the channel $\mcl{C}^*_{\povm{P},\gamma}$, which describes the observer's constraints, is spanned precisely by the projections in $\povm{\Pi}_{\povm{P},\gamma}$, and it is this algebra that encapsulates all information recoverable under those constraints.

Moreover, using Theorem~\ref{theorem:dual}, we explicitly determine the conditional expectation~\cite{umegaki1962conditional},~\cite[9.2 Conditional Expectations]{petz2007quantum} that characterizes the set of fixed points of the coarse-graining map.

\begin{theorem}
\label{theorem:rdm-coarse}
    Let $\povm{\Pi}_{\povm{P},\gm}=\{\Pi_y\}_{y\in\set{Y}}$ be the MPPP with respect to POVM $\povm{P}=\{P_x\}_{x\in\set{X}}$ and prior $\gm$. Then,
    \begin{align}
        \Delta_{\povm{P},\gm}(\bigcdot)&=\lim_{n\to\infty}\frac{1}{n}\sum_{k=1}^{n}\mcl{C}_{\povm{P},\gm}^k(\bigcdot)\;. \label{eq:DeltaCk}
    \end{align}
    Furthermore, for any quantum state $\rho$, 
    \begin{align}
        \Delta_{\povm{P},\gm}(\rho)=\rho
    \end{align}
    if and only if
    \begin{align}
        \mcl{C}_{\povm{P},\gm}(\rho)=\rho\;.
    \end{align}
\end{theorem}

The proofs in this subsection rely on foundational results in quantum information theory and operator algebras regarding quantum statistical sufficiency and conditional expectations, especially the theorems of Petz and Takesaki~\cite{petz1988sufficiency,takesaki1972conditional}. To apply these results to our context, we explicitly state our primary assumptions. First, we work within a finite-dimensional Hilbert space. Second, the prior state $\gm$ is invertible, or faithful. Finally, we assume that $\mcl{M}_{\povm{P}}(\gm)$ is invertible (i.e., faithful). These assumptions, particularly the condition that $\gm$ and $\mcl{M}_{\povm{P}}(\gm)$ are invertible, are necessary for applying Lemma~\ref{lemma:dual} below.

The proof of Theorem~\ref{theorem:dual} relies on the following result, which specializes~\cite[Theorem 2]{petz1988sufficiency} and \cite[Lemma 6.12]{hiai2021quantum} to the present situation.

\begin{lemma}
\label{lemma:dual}
    Let $\povm{P}=\{P_x\}_{x\in\set{X}}$ be a POVM and $\gm$ be an invertible quantum state such that $\mcl{M}_{\povm{P}}(\gm)$ is invertible. Let $B\in\sn\{\ketbra{x}\}_{x\in\set{X}}$. The following are then equivalent:
    \begin{enumerate}
        \item $(\mcl{R}_{\mcl{M}_{\povm{P}},\gm}^*\circ\mcl{M}_{\povm{P}}^*)(B)=B\;.$
        \item $\mcl{M}_{\povm{P}}^{*}(B^{\dagger}B)=\mcl{M}_{\povm{P}}^{*}(B)^{\dagger}\mcl{M}_{\povm{P}}^{*}(B)$ and $\;\mcl{M}_{\povm{P}}^*(B)\gm=\gamma\mcl{M}_{\povm{P}}^*(B)\;.$
    \end{enumerate}
    Furthermore, we get
    \begin{align}
        \mf{E}_{\mcl{C}_{\povm{P},\gm}^*}\coloneq\left\{B\in\sn\{\ketbra{x}\}_{x\in\set{X}}\mid(\mcl{R}_{\mcl{M}_{\povm{P}}, \gm}^*\circ\mcl{M}_{\povm{P}}^*)(B)=B\right\}\simeq\mf{F}_{\mcl{C}_{\povm{P},\gm}^*}\;,
    \end{align}
    where $\mcl{M}_{\povm{P}}^{*}$ is a $*$-isomorphism from $\mf{E}_{\mcl{C}_{\povm{P},\gm}^*}$ to $\mf{F}_{\mcl{C}_{\povm{P},\gm}^*}$, whose inverse is $\mcl{R}_{\mcl{M}_{\povm{P}},\gm}^*$.
\end{lemma}

\begin{proof}
    By applying results in Ref.~\cite[Theorem 2]{petz1988sufficiency}, \cite[Lemma 6.12]{hiai2021quantum}, in our setting, it follows that
    \begin{align}
        (\mcl{R}_{\mcl{M}_{\povm{P}},\gm}^*\circ\mcl{M}_{\povm{P}}^*)(B)=B
    \end{align}
    if and only if $\mcl{M}_{\povm{P}}^{*}(B^{\dagger}B)=\mcl{M}_{\povm{P}}^{*}(B)^{\dagger}\mcl{M}_{\povm{P}}^{*}(B)$ and
    \begin{align}
        \mcl{M}_{\povm{P}}^{*}\Big(\mcl{M}_{\povm{P}}(\gm)^{it}B\mcl{M}_{\povm{P}}(\gm)^{-it}\Big)=\gm^{it}\mcl{M}_{\povm{P}}^{*}(B)\gm^{-it}\;,
        \label{eq:MBBMBM}
    \end{align}
    for all $t\in\mbb{R}$. In our case, using the fact that the image of $\mcl{M}_{\povm{P}}$ (the domain of $\mcl{M}_{\povm{P}}^*$) is inside of $\sn\{\ketbra{x}\}_{x\in\set{X}}$, we obtain $\mcl{M}_{\povm{P}}(\gm)^{it}B\mcl{M}_{\povm{P}}(\gm)^{-it}=B$. Therefore, condition (\ref{eq:MBBMBM}) is equivalent to
    \begin{align}
        \mcl{M}_{\povm{P}}^{*}(B)=\gm^{it}\mcl{M}_{\povm{P}}^{*}(B)\gm^{-it}\;,\qquad\f t\in\mathbb{R}\;.\label{eq:MBGB}
    \end{align}
This means that  $[\mcl{M}_{\povm{P}}^{*}(B),\gm^{-it}]=0$ for $t\neq0$. Since $\gm$ is Hermitian and $\gm=(\gm^{it})^{1/it}$,  $[\mcl{M}_{\povm{P}}^*(B),\gm]=0$.
    Thus, condition~(\ref{eq:MBBMBM}) is equivalent to $\mcl{M}_{\povm{P}}^*(B)\gm=\gamma\mcl{M}_{\povm{P}}^*(B)$. The converse implication is trivial.

    Furthermore, as shown in~\cite[Theorem 2]{petz1988sufficiency},~\cite[Lemma 6.12]{hiai2021quantum}, $\mcl{M}_{\povm{P}}^*$ restricted onto $\mf{E}_{\mcl{C}_{\povm{P},\gm}^*}$ is a $*$-isomorphism to $\mf{F}_{\mcl{C}_{\povm{P},\gm}^*}$ whose inverse is $\mcl{R}_{\mcl{M}_{\povm{P}},\gm}^*$.
\end{proof}

\ 

We are now ready to prove Theorem~\ref{theorem:dual} and Theorem~\ref{theorem:rdm-coarse}.

\begin{proof}[Proof of Theorem~\ref{theorem:dual}]
    
    It is straightforward by calculation from (\ref{eq:CGMdfn}) and the definition of MPPP (Definition \ref{definition:mppp}) that $\mf{F}_{\mcl{C}_{\povm{P},\gm}^*}\supseteq\sn\left\{\Pi_y\right\}_{y\in\set{Y}}$.
    In the following, we prove the opposite relation $\mf{F}_{\mcl{C}_{\povm{P},\gm}^*}\subseteq\sn\{\Pi_y\}_{y\in\set{Y}}$. By Lemma~\ref{lemma:dual}, $\mcl{M}_{\povm{P}}^{*}$ is an isomorphism from $\mf{E}_{\mcl{C}_{\povm{P},\gm}^*}$ to $\mf{F}_{\mcl{C}_{\povm{P},\gm}^*}$, where
    \begin{align}
        \mf{E}_{\mcl{C}_{\povm{P},\gm}^*}=\left\{B\in\sn\{\ketbra{x}\}_{x\in\set{X}}\mid\mcl{M}_{\povm{P}}^*(B^{\dagger}B)=\mcl{M}_{\povm{P}}^*(B)^{\dagger}\mcl{M}_{\povm{P}}^*(B),\; \mcl{M}_{\povm{P}}^*(B)\gm=\gamma\mcl{M}_{\povm{P}}^*(B)\right\}\;.
        \label{eq:CeqBspan}
    \end{align}
    Our strategy will be to prove that
    \begin{align}
        \mf{E}_{\mcl{C}_{\povm{P},\gm}^*}\subseteq\sn\left\{\sum_{x\in\set{X}_y}\ketbra{x}\right\}_{y\in\set{Y}},
        \label{eq:ECeqSPANp}
    \end{align}
    where $\{\set{X}_y\}_{y\in\set{Y}}$ is the irreducible $\gamma$-disconnected partition of $\mathcal{X}$ (see Definition \ref{dfn:irreducibledisconnected}). Then, the statement follows. This is because, noting that $\mcl{M}_{\povm{P}}^{*}(\ketbra{x})=P_x$ and $\mcl{M}_{\povm{P}}^{*}(\sum_{x\in\set{X}_y}\ketbra{x})=\Pi_y$, we have
    \begin{align}
        \mf{F}_{\mcl{C}_{\povm{P},\gm}^*}
        =
        \mcl{M}_{\povm{P}}^*\left(\mf{E}_{\mcl{C}_{\povm{P},\gm}^*}\right)
        \subseteq\sn\left\{\Pi_y\right\}_{y\in\set{Y}}.
    \end{align}
    
    In order to prove (\ref{eq:ECeqSPANp}), it suffices to show that $\sum_{x\in\mathcal{X}}\alpha_x\ketbra{x}\in\mf{E}_{\mcl{C}_{\povm{P},\gm}^*}$ only if $\alpha_x=\alpha_{x'}$ for any $x,x'\in\mathcal{X}_y$ and any $y\in\mathcal{Y}$.
    We start with
    \begin{align}
        \sum_{x\in\set{X}}\alpha_x^*\alpha_xP_x=        \mcl{M}_{\povm{P}}^*\left[\sum_{x\in\set{X}}\alpha_x^*\alpha_x\ketbra{x}\right]
        =        \mcl{M}_{\povm{P}}^*\left[\left(\sum_{x\in\set{X}}\alpha_x\ketbra{x}\right)^{\dagger}\left(\sum_{x'\in\set{X}}\alpha_{x'}\ketbra{x'}\right)\right]
    \end{align}
and
\begin{align}    
        \left(\sum_{x\in\set{X}}\alpha_x^*P_x\right)\left(\sum_{x'\in\set{X}}\alpha_{x'}P_{x'}\right)
        =\mcl{M}_{\povm{P}}^*\left(\sum_{x\in\set{X}}\alpha_x\ketbra{x}\right)^{\dagger}\mcl{M}_{\povm{P}}^*\left(\sum_{x'\in\set{X}}\alpha_{x'}\ketbra{x'}\right).
    \end{align}
    From the first condition in (\ref{eq:CeqBspan}), the above two are equal. 
    Therefore, for $\xi\in\{\openone,\gamma\}$, it holds that
    \begin{align}
    \sum_{x,x'\in\set{X}}\alpha_x^*\alpha_xP_x\xi P_{x'}
        &=        \sum_{x\in\set{X}}\alpha_x^*\alpha_xP_x\xi
        =        \left(\sum_{x\in\set{X}}\alpha_x^*P_x\right)\left(\sum_{x'\in\set{X}}\alpha_{x'}P_{x'}\right)\xi,
    \end{align}
    where we made use of the normalization condition $\sum_{x'\in\set{X}}P_{x'}=\openone$. Due to the second condition in (\ref{eq:CeqBspan}), the last one in the above equality is equal to
    \begin{align}
       \left(\sum_{x\in\set{X}}\alpha_x^*P_x\right)\xi\left(\sum_{x'\in\set{X}}\alpha_{x'}P_{x'}\right)       =\sum_{x,x'\in\set{X}}\alpha_x^*\alpha_{x'}P_x\xi P_{x'}.
    \end{align}
    Hence, we have
    \begin{align}        \sum_{x,x'\in\set{X}}\alpha_x^*(\alpha_x-\alpha_{x'})P_x\xi P_{x'}
        =0\;.
    \end{align}    
The above, taking its trace on $\xi$, in turn yields
    \begin{align}        \sum_{x,x'\in\set{X}}\alpha_x^*(\alpha_x-\alpha_{x'}){\rm Tr}\left[\left(\xi^{\frac{1}{2}}P_x\xi^{\frac{1}{2}}\right)\left(\xi^{\frac{1}{2}} P_{x'}\xi^{\frac{1}{2}}\right)\right]
        =0.
    \end{align}
    Obviously, this relation holds when the roles of $x$ and $x'$ are exchanged. Summing up these two, we arrive at
    \begin{align}               
        \sum_{x,x'\in\set{X}}|\alpha_{x}-\alpha_{x'}|^2{\rm Tr}\left[\left(\xi^{\frac{1}{2}}P_x\xi^{\frac{1}{2}}\right)\left(\xi^{\frac{1}{2}} P_{x'}\xi^{\frac{1}{2}}\right)\right]=0.
        \label{eq:alphaxx}
    \end{align}
Note that $\xi^{\frac{1}{2}}P_x\xi^{\frac{1}{2}}\geq0$ and thus ${\rm Tr}[(\xi^{\frac{1}{2}}P_x\xi^{\frac{1}{2}})(\xi^{\frac{1}{2}} P_{x'}\xi^{\frac{1}{2}})]\geq0$, with the equality if and only if $(\xi^{\frac{1}{2}}P_x\xi^{\frac{1}{2}})(\xi^{\frac{1}{2}} P_{x'}\xi^{\frac{1}{2}})=0$. Because of the invertibility of $\xi$, this condition is equivalent to $P_x\xi P_{x'}=0$. Hence, (\ref{eq:alphaxx}) holds only if $\alpha_x=\alpha_{x'}$ for any $x$ and $x'$ such that $P_x\xi P_{x'}\neq0$. 
This argument applies for both of $\xi=\openone,\gamma$.
Thus, it must hold that $\alpha_x=\alpha_{x'}$ whenever $P_xP_{x'}\neq0$ or $P_x\gamma P_{x'}\neq0$, which implies that $\alpha_x$ must be constant in $\mathcal{X}_y$.   
\end{proof}

\begin{proof}[Proof of Theorem~\ref{theorem:rdm-coarse}]
We start by noticing that the map defined as
    \begin{align}
        \Delta_{\povm{P},\gm}^*(\bigcdot)=\sum_{y\in\set{Y}}\tr{(\bigcdot)\frac{\Pi_y\gm}{\tr{\Pi_y\gm}}}\Pi_y
    \end{align}
    is a \emph{conditional expectation}~\cite{umegaki1962conditional},~\cite[9.2 Conditional Expectations]{petz2007quantum} onto $\sn\left\{\Pi_y\right\}_{y\in\set{Y}}$, whereas 
    \begin{align}
        \mcl{E}^*(\bigcdot)=\lim_{n\to\infty}\frac{1}{n}\sum_{k=1}^{n}\left(\mcl{C}_{\povm{P},\gm}^*\right)^k(\bigcdot)
    \end{align}
    is a conditional expectation on $\mf{F}_{\mcl{C}_{\povm{P},\gm}^*}=\left\{A\in\mcl{L}(\sH)\mid\mcl{C}_{\povm{P},\gm}^*(A)=A\right\}$~\cite[Lemma 11]{hayden2004structure}.
    
    Let $\varphi_\gm$ be defined as
    \begin{align}
        \varphi_\gm:A\mapsto\tr{\gm A}\;.
    \end{align}
    Then, we get $\varphi_\gm\circ\Delta_{\povm{P},\gm}^*=\varphi_\gm$ and $\varphi_\gm\circ\mcl{E}^*=\varphi_\gm$. Thus, from Theorem~\ref{theorem:dual} and the uniqueness of $\varphi_\gm$-preserving conditional expectation (or \textit{Takesaki theorem})~\cite{takesaki1972conditional},~\cite[Theorem A.10]{hiai2021quantum},
    \begin{align}
        \Delta_{\povm{P},\gm}^*(\bigcdot)=\lim_{n\to\infty}\frac{1}{n}\sum_{k=1}^{n}(\mcl{C}_{\povm{P},\gm}^*)^k(\bigcdot)\;,
    \end{align}
    i.e., 
    \begin{align}
        \Delta_{\povm{P},\gm}(\bigcdot)=\lim_{n\to\infty}\frac{1}{n}\sum_{k=1}^{n}\mcl{C}_{\povm{P},\gm}^k(\bigcdot)\;.
    \end{align}
    Furthermore, for any quantum state $\rho$, we have~\cite[Proof of Theorem 5]{jencova2017preservation},~\cite[Example 9.4.]{petz2007quantum}
    \begin{align}
        \lim_{n\to\infty}\frac{1}{n}\sum_{k=1}^{n}\mcl{C}_{\povm{P},\gm}^k(\rho)=\rho
    \end{align}
    if and only if
    \begin{align}
        \mcl{C}_{\povm{P},\gm}(\rho)=\rho\;.
    \end{align}
\end{proof}

\section{Resource theory of microscopicity}
\label{section:microscopicity}

Having established the structure of macroscopic states, we now ask: what makes a state ``more or less macroscopic/microscopic'' than another? More crucially, how does this microscopicity relate to the visibility of quantum correlations under restricted observations? To answer this, we develop a resource theory of microscopicity, which captures observer-dependent irreversibility and correlation loss.

As briefly mentioned after Theorem~\ref{theorem:dual}, the MPPP associated to a measurement $\povm{P}$ and a prior state $\gamma$ encapsulates the information-preserving structure induced by both the measurement and the prior: all macroscopic states are block-diagonal with respect to it. The MPPP thus plays the role of an \textit{inferential reference frame} generalizing the ``choice of a preferred basis'' in the resource theory of coherence, where incoherent states are diagonal in a fixed basis. The difference is that, here, the inferential structure is shaped not only by the measurement but also by the prior $\gamma$. To make this notion precise, we construct a resource theory of \textit{microscopicity}, in which macroscopic states are treated as free states, and introduce the \textit{relative entropy of microscopicity} as a measure of irretrodictability, i.e., the extent to which microscopic details are lost and unrecoverable from macroscopic observations~\cite{buscemi2022observational,nagasawa2024generic}.

\subsection{Free states and operations}
\label{subsection:microscopicity}

In the framework of \textit{resource theories}~\cite{Chitamber2019quantum,gour2019how,gour2024resources}, a resource theory can be constructed either from the set of free states (i.e., the so-called ``geometric approach''~\cite{Zhou_2020}) or from the set of free operations (the ``operational approach''). In either case, however, it is important that the minimal conditions for internal consistency are satisfied: (i) every free state is mapped to a free state by a free operation, and (ii) every free state can be generated by a free operation.
A resource theory is said to be {\it convex} if the set of free states and/or the set of free operations is closed under convex combination.
A convex resource theory is further said to be {\it affine} if no resource state can be represented as a linear combination of free states~\cite[9.3 Affine Resource Theories]{gour2024resources}.

The set of free states in our resource theory of microscopicity is the set of macroscopic states, as characterized in Theorem \ref{theorem:macroscopic}.
Given a pair $(\povm{P}, \gamma)$, we denote such a set as
\begin{align}
    \mf{M}(\povm{P},\gm)&\coloneq\{\sigma\in\mcl{S}(\sH)\mid\Delta_{\povm{P},\gm}(\sigma)=\sigma\}\nonumber\\
    &=\left\{\left.\sum_{y\in\set{Y}}p_y\frac{\Pi_y\gm}{\tr{\Pi_y\gm}}\right|p_y\ge 0,\;\sum_yp_y=1\right\}
    \label{eq:freestates}\\
    &=\operatorname{span}\{\Pi_y\gm\}_y\cap\mcl{S}(\mcl{H})\;,\nonumber
\end{align}
where $\povm{\Pi}_{\povm{P},\gamma}=\{\Pi_y\}_y$ is the uniquely defined $\gamma$-commuting MPPP and $\Delta_{\povm{P},\gm}$ is the map defined by Eq.~(\ref{eq:RDMrtm}). The final equality in Eq.~(\ref{eq:freestates}) follows because all blocks $\Pi_y\gm$ are non-zero and orthogonal to each other by construction, so that any linear combination $\sum_{y\in\set{Y}}\lambda_y\Pi_y\gm$, with $\lambda_y\in\mathbb{R}$, belongs to $\mcl{S}(\mcl{H})$ if and only if $\lambda_y\ge 0$ and $\sum_y\lambda_y\tr{\Pi_y\gm}=1$, i.e., it is in fact a convex combination of the extremal points. Alternatively, although more implicitly, affinity of $\mf{M}(\povm{P},\gm)$ can be seen as a consequence of its definition, where $\sigma\in\mf{M}(\povm{P},\gm)$ if and only if $\Delta_{\povm{P},\gm}(\sigma)=\sigma$, where $\Delta_{\povm{P},\gm}$ is linear. Anyway, in conclusion, the resource theory of microscopicity is convex and affine.

Since for any $\rho\in\mcl{S}(\mcl{H})$ we have $\Delta_{\povm{P},\gm}(\rho)\in\mf{M}(\povm{P}, \gm)$, while $\Delta_{\povm{P},\gm}(\sigma)=\sigma$ for any $\sigma\in\mf{M}(\povm{P}, \gm)$, the map $\Delta_{\povm{P},\gm}$ is \emph{the resource destroying map} (RDM)~\cite{liu2017resource} in the resource theory of microscopicity. From it, free operations (i.e., \textit{macroscopic operations} in our language) can be introduced as follows.

\begin{definition}[Macroscopic operations]
    We introduce three classes of free operations.
A CPTP map $\mcl{E}$ is called a 
\begin{itemize}
    \item \emph{microscopicity non-generating operation} (MNO) if $\sigma\in\mf{M}(\povm{P}, \gm)$ $\implies$ $\mcl{E}(\sigma)\in\mf{M}(\povm{P},\gm)$;
    \item \emph{resource-destroying covariant operation} (RCO) if $\mcl{E}\circ\Delta_{\povm{P},\gm}=\Delta_{\povm{P},\gm}\circ\mcl{E}$;
    \item \emph{coarse-graining covariant operation} (CCO) if $\mcl{E}\circ\mcl{C}_{\povm{P},\gm}=\mcl{C}_{\povm{P},\gm}\circ\mcl{E}$.
\end{itemize}

\end{definition}

MNO and RCO are natural generalizations of maximally incoherent operations (MIO) and dephasing-covariant incoherent operations (DIO), respectively, in the resource theory of coherence~\cite{aberg2006quantifying,chitambar2016critical,streltsov2017colloquium} (see Example~\ref{example:coherence} below). CCO instead constitutes a new class, since in general $\mcl{C}_{\povm{P},\gm}\neq \Delta_{\povm{P},\gm}$. The following proposition shows that the above classes of macroscopic operations in fact form a hierarchy, which can be strict or not, depending on the POVM $\povm{P}$ and the prior $\gamma$.

\begin{proposition}\label{proposition:CCOMNO}
The three classes of free operations satisfy the following inclusion relation:
\begin{align}
    \mr{CCO}\subseteq\mr{RCO}\subseteq\mr{MNO}\;.
    \label{eq:freeOPhie}
\end{align}
A sufficient condition for the first equality is $\povm{P}\preceq \povm{\Pi}_{\povm{P},\gamma}$.
The second equality holds if and only if the MPPP is trivial, i.e. $\povm{\Pi}_{\povm{P},\gamma}=\{\openone\}$.
\end{proposition}

\begin{proof}
    The first relation follows from the fact that $\Delta_{\povm{P},\gm}$ is represented by $\mcl{C}_{\povm{P},\gm}$ as Eq.~(\ref{eq:DeltaCk}). The second relation follows because $\mcl{E}(\sigma)=\mcl{E}\circ\Delta_{\povm{P},\gm}(\sigma)=\Delta_{\povm{P},\gm}\circ\mcl{E}(\sigma)\in\mf{M}(\povm{P}, \gm)$ for any $\sigma\in\mf{M}(\povm{P}, \gm)$ and $\mathcal{E}\in\mr{RCO}$.
    
    The relation $\povm{P}\preceq \povm{\Pi}_{\povm{P},\gamma}$, which is equivalent to $\povm{P}\simeq \povm{\Pi}_{\povm{P},\gamma}$ as $\povm{P}\succeq \povm{\Pi}_{\povm{P},\gamma}$ holds by definition, guarantees that any element of $\povm{P}$ is proportional to an element of $\povm{\Pi}_{\povm{P},\gamma}$, which is the case when $\povm{P}$ is obtained from $\povm{\Pi}_{\povm{P},\gamma}$ by probabilistic splitting of the measurement outcomes. 
    Hence, if it holds, we have $\mathcal{C}_{\povm{P},\gm}=\Delta_{\povm{P},\gm}$ and thus $\mr{CCO}=\mr{RCO}$.
    
    When $\Pi_{\povm{P},\gamma}=\{\openone\}$, the prior state $\gamma$ is the only macroscopic state due to (\ref{eq:freestates}). In this case, RDM takes the form of $\Delta_{\povm{P},\gm}(\:\bigcdot\:)={\rm Tr}[\:\bigcdot\:]\gamma$ and is the only element of MNO, which implies $\mr{RCO}=\mr{MNO}$. Conversely, suppose that $\povm{\Pi}_{\povm{P},\gm}=\{\Pi_y\}_{y\in\mathcal{Y}}$, where $|\mathcal{Y}|\geq2$. Let $|\phi\rangle$ be a normalized vector such that $\Pi_{y_1}|\phi\rangle\neq0$ and $\Pi_{y_2}|\phi\rangle\neq0$ for some $y_1\neq y_2$, and define an operation
    \begin{equation}
        \mathcal{E}(\bigcdot):={\rm Tr}\Big[(\bigcdot)\;|\phi\rangle\!\langle\phi|\Big]\Pi_{y_1}\gamma+{\rm Tr}\Big[(\bigcdot)\;(I-|\phi\rangle\!\langle\phi|)\Big]\Pi_{y_2}\gamma\;.
    \end{equation}
    It is straightforward that $\Delta_{\povm{P},\gm}\circ\mathcal{E}(|\phi\rangle\!\langle\phi|)=\Pi_{y_1}\gamma\neq\mathcal{E}\circ\Delta_{\povm{P},\gm}(|\phi\rangle\!\langle\phi|)$ and $\Delta_{\povm{P},\gm}\circ\mathcal{E}=\mathcal{E}$, which implies $\mr{RCO}\subsetneq\mr{MNO}$.
\end{proof}

\

\begin{remark}
In general, the first inclusion, Eq.~\eqref{eq:freeOPhie} can also be strict, that is, it could indeed be the case that $\mr{CCO}\subsetneq\mr{RCO}$. An example is given as follows.
Consider a one-qubit system ($\dim\mathcal{H}=2$) and let $\povm{P}\equiv\{P_0,P_1\}$ be a binary POVM such that
\begin{align}
    P_0=\frac{2}{3}|0\rangle\!\langle0|+\frac{1}{3}|1\rangle\!\langle1|,
    \quad
    P_1=\frac{1}{3}|0\rangle\!\langle0|+\frac{2}{3}|1\rangle\!\langle1|,
\end{align}
for which the MPPP is trivial.
For the uniform prior, the RDM takes the form of $\Delta_{\povm{P},u}(\:\bigcdot\:)={\rm Tr}[\:\bigcdot\:]u$, thus commutes with the Hadamard gate $H:=|+\rangle\!\langle0|+|-\rangle\!\langle1|$.
Nevertheless, the coarse-graining map $\mathcal{C}_{\povm{P},u}(\bigcdot)$ does not commute with $H$. The unitary channel $H$ is thus in RCO but not in CCO. The search for necessary and equivalent conditions for equality between CCO and RCO is left open for future research.
\end{remark}

\subsection{Relative Entropy of Microscopicity}

We introduce the {\it relative entropy of microscopicity} to quantify how a state is distant from the set of the macroscopic states, and analyze its property as a resource measure. 
This is a special case of the Umegaki relative entropy of resourcefulness in general resource theories~\cite{gour2024resources}.

\begin{definition}[Relative entropy of microscopicity]
    For a POVM $\povm{P}=\{P_x\}_{x}$ and an invertible quantum state $\gm$, the relative entropy of microscopicity is defined as
    \begin{align}
        D_{\povm{P},\gm}(\bigcdot)\coloneq \inf_{\sigma\in\mf{M}(\povm{P}, \gm)}D(\bigcdot\|\sigma).
    \end{align}
    
\end{definition}
\noindent 
By definition and from a general argument, it is straightforward to see that $D_{\povm{P},\gm}(\bigcdot)$ is a microscopicity monotone, that is, $D_{\povm{P},\gm}(\rho)\geq D_{\povm{P},\gm}(\mathcal{E}(\rho))$ for any $\rho\in\mathcal{S}(\mathcal{H})$ and $\mathcal{E}\in\mr{MNO}$.
The relative entropy of microscopicity can be represented in terms of RDM and is bounded below by the observational deficit:
\begin{theorem}
\label{theorem:micro}
    Let $\povm{\Pi}_{\povm{P},\gm}\equiv\povm{\Pi}=\{\Pi_y\}_{y\in\set{Y}}$ be the MPPP with respect to a quantum prior $\gm$ and a POVM $\povm{P}=\{P_x\}_{x\in\set{X}}$. Then, for any quantum state $\rho$,
    \begin{align}
        D_{\povm{P},\gm}(\rho)&=D(\rho\|\Delta_{\povm{P},\gm}(\rho))\\
        %&=\delta_{\povm{\Pi}_{\povm{P},\gm}}(\rho\|\sigma)\\
        &\geq\delta_{\povm{P}}(\rho\|\sigma)\label{eq:last-ineq}
    \end{align}
    for all $\sigma\in\mf{M}(\povm{P}, \gm)$. The condition for equality in~\eqref{eq:last-ineq} is that a classical post-processing that converts $\povm{P}$ to $\povm{\Pi}_{\povm{P},\gm}$, which is represented by the conditional probability distribution $\{p(y|x)\}$ as $\Pi_y=\sum_{x\in\mathcal{X}}p(y|x)P_x$, is sufficient with respect to the dichotomy $(\mcl{M}_{\povm{P}}(\rho),\mcl{M}_{\povm{P}}(\sigma))$.
\end{theorem}
\begin{proof}
    For any $\sigma\in\mf{M}(\povm{P}, \gm)$, we get
    \begin{align}
        &D(\rho\|\sigma)\\
        &=\tr{\rho\log\rho}-\tr{\rho\log\sigma}\\
        &
            =\tr{\rho\log\rho}-\tr{\rho\log\Delta_{\povm{P},\gm}(\rho)}+\tr{\rho\log\Delta_{\povm{P},\gm}(\rho)}-\tr{\rho\log\Delta_{\povm{P},\gm}(\sigma)}
        \\
        &
            =D(\rho\|\Delta_{\povm{P},\gm}(\rho))+\tr{\mathcal{P}_{\povm{\Pi}}(\rho)\log\Delta_{\povm{P},\gm}(\rho)}-\tr{\mathcal{P}_{\povm{\Pi}}(\rho)\log\Delta_{\povm{P},\gm}(\sigma)}\;,\label{eq:81}
    \end{align}
    where we introduced the pinching map $\mathcal{P}_{\povm{\Pi}}(\bigcdot)\coloneq\sum_{y\in\set{Y}}\Pi_y\bigcdot\Pi_y$. (Note that, in general, $\mathcal{P}_{\povm{\Pi}}\neq \Delta_{\povm{P},\gm}$.)
    The second term is calculated as
    \begin{align}
        &\tr{\mathcal{P}_{\povm{\Pi}}(\rho)\log\Delta_{\povm{P},\gm}(\rho)}\nonumber\\
        &=\sum_y\tr{\Pi_y\rho\Pi_y\log\left(\tr{\Pi_y\rho}\frac{\Pi_y\gm}{\tr{\Pi_y\gm}}\right)}\\
        &=\sum_y\tr{\Pi_y\rho\Pi_y\left\{\log\left(\tr{\Pi_y\rho}\right)\openone+\log\left(\frac{\Pi_y\gm}{\tr{\Pi_y\gm}}\right)\right\}}.
    \end{align}
    Similarly, for the third term we have
    \begin{align}
        &\tr{\mathcal{P}_{\povm{\Pi}}(\rho)\log\Delta_{\povm{P},\gm}(\sigma)}\nonumber\\
        &=\sum_y\tr{\Pi_y\rho\Pi_y\left\{\log\left(\tr{\Pi_y\sigma}\right)\openone+\log\left(\frac{\Pi_y\gm}{\tr{\Pi_y\gm}}\right)\right\}}.
    \end{align}
    Thus,
    \begin{align}
        &\tr{\mathcal{P}_{\povm{\Pi}}(\rho)\log\Delta_{\povm{P},\gm}(\rho)}-\tr{\mathcal{P}_{\povm{\Pi}}(\rho)\log\Delta_{\povm{P},\gm}(\sigma)}\nonumber\\
        &=D(\mcl{M}_{\povm{\Pi}_{\povm{P},\gm}}(\rho)\|\mcl{M}_{\povm{\Pi}_{\povm{P},\gm}}(\sigma)).
    \end{align}
    Combining this with (\ref{eq:81}), we obtain
    \begin{align}
        D(\rho\|\Delta_{\povm{P},\gm}(\rho))&=D(\rho\|\sigma)-D(\mcl{M}_{\povm{\Pi}_{\povm{P},\gm}}(\rho)\|\mcl{M}_{\povm{\Pi}_{\povm{P},\gm}}(\sigma))\label{eq:86}
        %&=\delta_{\povm{\Pi}_{\povm{P},\gm}}(\rho\|\sigma).
    \end{align}
    Since $\povm{\Pi}_{\povm{P},\gm}$ is obtained from $\povm{P}$ by classical post-processing, we have
    \begin{align}
        
        D(\mcl{M}_{\povm{P}}(\rho)\|\mcl{M}_{\povm{P}}(\sigma))\geq D(\mcl{M}_{\povm{\Pi}_{\povm{P},\gm}}(\rho)\|\mcl{M}_{\povm{\Pi}_{\povm{P},\gm}}(\sigma))
    \end{align}
    by the data processing inequality. Thus, we obtain $D(\rho\|\Delta_{\povm{P},\gm}(\rho))\geq\delta_{\povm{\P}}(\rho\|\sigma)$.
    From (\ref{eq:86}), and noting that $\Delta_{\povm{P},\gm}(\rho)\in\mf{M}(\povm{P}, \gm)$, we also have
    \begin{align}
        D(\rho\|\Delta_{\povm{P},\gm}(\rho))\geq \inf_{\sigma\in\mf{M}(\povm{P}, \gm)}D(\rho\|\sigma)\geq D(\rho\|\Delta_{\povm{P},\gm}(\rho)),
    \end{align}
    which implies $D_{\povm{P},\gm}(\rho)=D(\rho\|\Delta_{\povm{P},\gm}(\rho))$.
\end{proof}

\subsection{Reduction to other resource theories}
\label{subsection:examples}
The resource theories of \emph{coherence}, \emph{athermality}, \emph{nonuniformity} and \emph{asymmetry} are obtained from the resource theory of microscopicity by properly choosing the prior state $\gamma$, or by restricting the POVM $\povm{P}$ so that the MPPP satisfies certain conditions, or both.

\begin{figure}[h]
    \centering
    \includegraphics[width=0.6\linewidth]{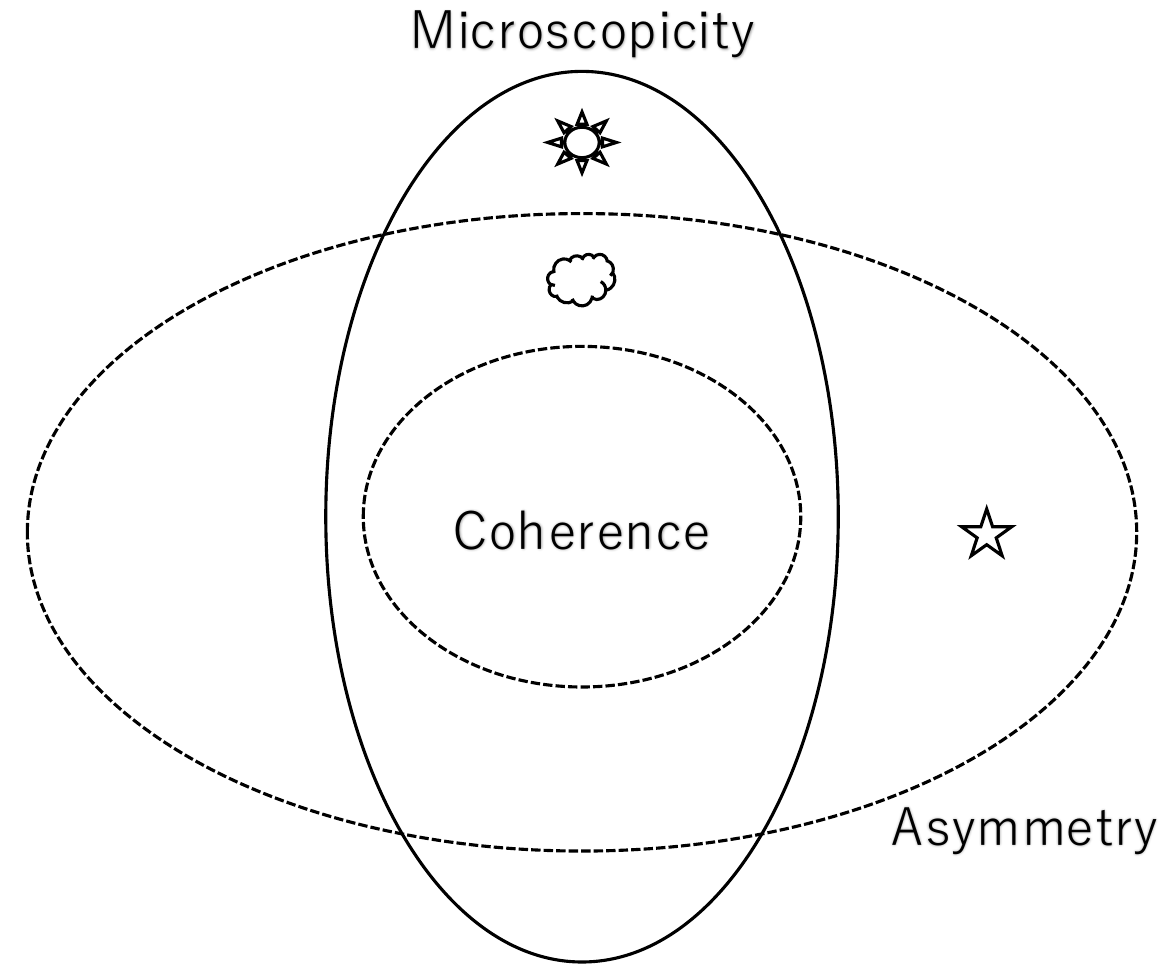}
    \caption{A schematic diagram illustrating the relationship between the resource theory of microscopicity and other prominent resource theories. The large vertical ellipse represents the theory of microscopicity, which encompasses several other theories. Within this framework, the inner circle signifies the resource theory of coherence, while the sun symbol represents the resource theory of athermality, a specific case of microscopicity (see Example~\ref{example:athermality}). The dashed, partially overlapping horizontal ellipse depicts the resource theory of asymmetry, where the cloud denotes a special case in which the multiplicity space of group representations is trivial (see Example~\ref{example:asymmetry}). Finally, the star symbol stands for the resource theory of block-coherence, which cannot be derived from our framework due to the more constrained structure of macroscopic states (see Remark~\ref{remark:block}).}
\end{figure}

\begin{example}[Resource theory of coherence]
\label{example:coherence}
    Fix a complete orthonormal system (CONS) $\{\ket{i}\}_i$ on the $d$-dimensional Hilbert space $\mcl{H}$. 
    Suppose that the prior state $\gamma$ is diagonal in that basis and that $\povm{P}$ is a rank-one POVM which is also diagonal. In this case, the MPPP is the rank-one PVM in that basis. The RDM reduces to the completely dephasing (pinching) map $\mathcal{P}_{\{\ketbra{i}\}_{i}}(\bigcdot)=\sum_i \braket{i|\bigcdot|i}\ketbra{i}$. 
    The free states are the diagonal states:
    \begin{align}
        \mf{M}(\povm{P}, \gm)&= \left\{\sum_ip_i\ketbra{i}\middle|p_i\geq 0,\sum_i p_i=1\right\}\\
        &=\operatorname{span}\{\ketbra{i}\}_i\cap\mcl{S}(\mcl{H})\;.\nonumber
    \end{align}
    The MNO and RCO reduces to the maximally incoherent operations (MIO) and the dephasing-covariant operations (DIO)~\cite{aberg2006quantifying,chitambar2016critical,streltsov2017colloquium}, respectively, defined as 
    \begin{align}
        \mr{MIO}&\coloneq\{\mcl{E}\in\mr{CPTP}\mid \sigma\in\mf{M}(\povm{P}, \gm)\implies\mcl{E}(\sigma)\in\mf{M}(\povm{P},\gm)\}\\
        \mr{DIO}&\coloneq \{\mcl{E}\in\mr{CPTP}\mid\mathcal{P}_{\{\ketbra{i}\}_{i}}\circ\mcl{E}=\mcl{E}\circ\mathcal{P}_{\{\ketbra{i}\}_{i}}\}\;.
    \end{align}
    It is known that the inclusion of DIO in MIO is strict, i.e., $\mr{DIO}\subsetneq\mr{MIO}$~\cite{chitambar2016critical}. This is also valid in the theory of microscopicity, since the inclusion of RCO in MNO is strict as soon as the MPPP is not the trivial one, see Proposition~\ref{proposition:CCOMNO}.
\end{example}

\begin{example}[Resource theory of athermality and nonuniformity]
\label{example:athermality}
    Suppose that the MPPP is trivial, i.e. $\povm{\Pi}_{\povm{P},\gm}=\{\openone\}$. This happens when, for example, ${\rm Tr}[P_xP_{x'}]\neq0$ for any $x,x'\in\mathcal{X}$. Due to Condition IV in Theorem \ref{theorem:macroscopic}, the only macroscopic state in this case is the prior state $\gamma$. Furthermore, the classes of free operations $\mr{RCO}$ coincides with $\mr{MNO}$ (Proposition \ref{proposition:CCOMNO}). If $\gamma$ is the Gibbs state for a given Hamiltonian and a temperature, it reduces to the resource theory of athermality~\cite{brandao2013resource,brandao2015second}. The free operations are the {\it Gibbs preserving operations}, which are the operations that preserves the Gibbs state invariant. If $\gamma$ is the maximally mixed state, it is the resource theory of nonuniformity, in which the free operations are those that are described as the unital CPTP channels.
\end{example}

\begin{example}[Resource theory of asymmetry]
\label{example:asymmetry}
    Let $G$ be a group that represents a symmetry of a given quantum system, and let $\mathcal{U}_G:=\{U_g\}_{g\in G}$ be a unitary representation of $G$.  
    In the resource theory of asymmetry, a state $\sigma\in\mathcal{S}(\mathcal{H})$ is said to be {\it symmetric} (or {\it $G$-invariant}) if it is invariant under the action of any element of $G$, that is, $U_g\sigma U_g^\dagger=\sigma$ for all $g\in G$. 
    It is known that, given a unitary representation $\mathcal{U}_G$, there exists a decomposition of the Hilbert space $\mathcal{H}$ into the direct-sum-product form $\sH=\bigoplus_{\lam\in\mr{Irr}(U)}A_\lam=\bigoplus_{\lam\in\mr{Irr}(U)}B_\lam\otimes C_\lam$, such that $U_g$ is decomposed into $U_g=\bigoplus_{\lam\in\mr{Irr}(U)}v_{g,\lambda}^{B_\lam}\otimes I^{C_\lam}$, where each $v_{g,\lambda}$ is an irreducible unitary representation of $G$. A state $\sigma\in\mathcal{S}(\mathcal{H})$ is symmetric if and only if it is of the form
    \begin{align}
        \sigma=\bigoplus_{\lam\in\mr{Irr}(U)}u^{B_\lam}\otimes\sigma_\lam^{C_\lam},
    \end{align}
    where $\sigma_\lam^{C_\lam}\coloneq\mr{Tr}_{B_\lam}[\Pi^{A_\lam}\sigma\Pi^{A_\lam}]$.
    We particularly consider the case where the multiplicity space is trivial, i.e. $\dim{C_\lambda}=1$. Then the above decomposition yields $\sigma=\bigoplus_{\lam\in\mr{Irr}(U)}\tilde{c}_{\sigma,\lambda} u^{B_\lam}$, with the coefficients satisfying $\tilde{c}_{\sigma,\lambda}\geq0$ and $\sum_{\lam\in\mr{Irr}(U)}\tilde{c}_{\sigma,\lambda}=1$. The states of this form are exactly the macroscopic states for the uniform prior state, with the block-diagonal decomposition of the space defined by the MPPP (see IV in Theorem \ref{theorem:macroscopic}). When $G$ is a compact Lie group, there is a unique group invariant measure $\mu$ on $\mathcal{U}_G$, which is referred to as the Haar measure. 
    The $G$-twirling operation~\cite[15.2.1 The G-Twirling Operation]{gour2024resources} is then defined by
    \begin{align}
        \mathcal{T}_{\mathcal{U}_G}(\bigcdot):=\int_{g\in G} U_g(\bigcdot)U_g^{\dagger}\;\mu(dg),
    \end{align}
    and is equal to the RDM $\Delta_{\povm{P},u}(\bigcdot)$.
    Let
    \begin{align}
        \mf{C}(G)\coloneq\left\{\mcl{E}\in\mr{CPTP}\;\middle|\;\mcl{E}(U_g\;\bigcdot\;U_g^{\dagger})=U_g\;\mcl{E}(\bigcdot)\;U_g^{\dagger}\;(\f g\in G)\right\}
    \end{align}
    be the set of covariant operations. In this case, $\mf{C}(G)\subset\mr{RCO}$. 
    It is left open whether $\mf{C}(G)=\mr{CCO}$ holds for a proper choice of the POVM.
\end{example}

\begin{remark}[Resource theory of block-coherence]
\label{remark:block}
    It is obvious from Condition IV in Theorem \ref{theorem:macroscopic} that the macroscopic states are not only block-diagonal with respect to the MPPP but also in a fixed state in each block, which is $\Pi_{\povm{P},\gamma}\gamma$.
    Thus, a state which is block-diagonal in MPPP is not necessarily a macroscopic state.
    Consequently, the resource theory of block coherence~\cite{bischof2019resource} cannot be deduced from the resource theory of microscopicity. 
\end{remark}

\section{Observer-dependent measure of quantum correlations}
\label{section:correlation}

This section applies the framework developed in the preceding sections to explore quantum correlations, specifically \emph{entanglement}~\cite{einstein1935can,horodecki2009quantum}, \emph{deficit}~\cite{oppenheim2002thermodynamical,horodecki2003local,devetak2005distillation,horodecki2005local}, and \emph{discord}~\cite{ollivier2001quantum,modi2012classical,ming2018quantum}, from the perspective of the observer's \emph{(quantum) reference frame}~\cite{bartlett2007reference,gour2008resource,fewster2025quantum}. Our approach is inspired by Ryszard Horodecki’s vision of correlations as physical resources, in which correlations emerge not only from state structure but also from the interplay between observation, inference, and prior knowledge.

More precisely, here we explore how limitations in measurement capabilities influence the perception and utility of these correlations. We establish a resource theory for locally macroscopic states and their corresponding operations, considering observational limitations on a subsystem rather than the entire system. This leads us to introduce the concept of \emph{local microscopicity}, which describes quantum features that are lost when a system is observed through a coarse-grained local measurement.

\subsection{Locally macroscopic states and operations}
\label{subsection:locally-macro}

We begin with a complete characterization of states that are \emph{locally macroscopic} with respect to POVM $\povm{P}_A=\{P^x_A\}_x$ and prior state $\gm_{AB}=\gamma_A\otimes\gamma_B$. In what follows, $\mcl{M}_{\povm{P}_A}$ is the quantum-classical channel corresponding to the POVM $\povm{P}_A$ on the Hilbert space $\sH_A$, as in Eq.~\eqref{eq:qc-channel}.

\begin{theorem}[Locally macroscopic states]
\label{theorem:locally-micro-states}
    Let $\gm_{AB}\coloneq\gm_A\ot\gm_B$ be a product quantum prior in $\sH_A\ot\sH_B$, $\povm{\Pi}_{\povm{P}_A,\gm_A}\equiv\povm{\Pi}=\{\Pi_y\}_{y\in\set{Y}}$ be the MPPP with respect to POVM $\povm{P}_A$ and prior state $\gm_A$, and $\mcl{I}_B$ be an identity operator on the set of linear operators on a Hilbert space $\sH_B$. Then, for any $\rho_{AB}\in\mcl{S}(\sH_A\ot\sH_B)$, the following conditions are equivalent:
    \begin{enumerate}
        \item $D(\rho_{AB}\|\gm_{AB})=D\Big((\mcl{M}_{\povm{P}_A}\ot\mcl{I}_B)(\rho_{AB})\|(\mcl{M}_{\povm{P}_A}\ot\mcl{I}_B)(\gm_{AB})\Big)$;\label{theorem:locally-micro-states:divergence}
        \item $(\mcl{C}_{\povm{P}_A,\gm_A}\ot\mcl{I}_B)(\rho_{AB})\coloneq\Big[(\mcl{R}_{\mcl{M}_{{\povm{P}}_A},\;\gm_{A}}\circ\mcl{M}_{\povm{P}_A})\ot\mcl{I}_B\Big](\rho_{AB})=\rho_{AB}$;\label{theorem:locally-micro-states:coarse-graining}
        \item $(\Delta_{\povm{P}_A,\gm_A}\ot\mcl{I}_B)(\rho_{AB})=\rho_{AB}$;\label{theorem:locally-micro-states:resource-destroying}
        \item there exist $c_y^i\geq0$, $\lam_i\geq0$ satisfying $\sum_i\lam_i=1$, $\sigma_A^i\in\mcl{S}(\sH_A)$, and $\sigma_B^i\in\mcl{S}(\sH_B)$ such that
        \begin{align}
            \rho_{AB}&=\sum_{i}\lam_i\;\sigma_A^i\ot\sigma_B^i\\
            &=\sum_{i}\lam_i\;\Delta_{\povm{P}_A,\gm_A}(\sigma_A^i)\ot\sigma_B^i\\
            &=\sum_{i}\lam_i\left(\sum_{y\in\set{Y}}c_y^i\Pi_y\gm_A\right)\ot\sigma_B^i\;.
        \end{align}\label{theorem:locally-micro-states:explicit}
    \end{enumerate}
\end{theorem}
\begin{proof}
    The equivalence~(\ref{theorem:locally-micro-states:divergence})$\iff$(\ref{theorem:locally-micro-states:coarse-graining}) follows directly from the equality condition for the data processing inequality (DPI) of the Umegaki relative entropy, i.e., $D(\rho_{AB}\|\gm_{AB})=D((\mcl{M}_{\povm{P}_A}\ot\mcl{I}_B)(\rho_{AB})\|(\mcl{M}_{\povm{P}_A}\ot\mcl{I}_B)(\gm_{AB})) \iff\mcl{C}_{\povm{P}_A,\gm_{AB}}(\rho_{AB})\coloneq(\mcl{R}_{\mcl{M}_{\povm{P}_A}\ot\mcl{I}_B,\;\gm_{AB}})\circ(\mcl{M}_{\povm{P}_A}\ot\mcl{I}_B)(\rho_{AB})=\rho_{AB}$. By assumption $\gm_{AB}=\gm_A\ot\gm_B$ which together with the locality of measurement, implies $\mcl{C}_{\povm{P}_A,\gm_{AB}}=\mcl{C}_{\povm{P}_A,\gm_A}\ot\mcl{I}_B$. Furthermore,~(\ref{theorem:locally-micro-states:resource-destroying})$\iff$(\ref{theorem:locally-micro-states:explicit}) follows from $\Delta_{\povm{P}_A,\gm}$ being entanglement breaking~\cite{horodecki2003entanglement}, i.e., $(\Delta_{\povm{P}_A,\gm_A}\ot\mcl{I}_B)(\rho_{AB})$ is the separable state. Therefore, we show~(\ref{theorem:locally-micro-states:coarse-graining})$\iff$(\ref{theorem:locally-micro-states:resource-destroying}). By $\mcl{C}_{\povm{P}_A,\gm_A}$ also being entanglement breaking, if $(\mcl{C}_{\povm{P}_A,\gm_A}\ot\mcl{I}_B)(\rho_{AB})=\rho_{AB}$, then
    \begin{align}
        \rho_{AB}&=\sum_{i}\lam_i\sigma_A^i\ot\sigma_B^i\\
        &=(\mcl{C}_{\povm{P}_A,\gm_A}\ot\mcl{I}_B)\left(\sum_{i}\lam_i\sigma_A^i\ot\sigma_B^i\right)\\
        &=\sum_{i}\lam_i\mcl{C}_{\povm{P}_A,\gm_A}(\sigma_A^i)\ot\sigma_B^i\\
        &=\sum_{i}\lam_i\mcl{C}_{\povm{P}_A,\gm_A}^k(\sigma_A^i)\ot\sigma_B^i\;,
    \end{align}
    where $\sigma_A^i\in\mcl{S}(\sH_A)$ and $\sigma_B^i\in\mcl{S}(\sH_B)$. Moreover, from Theorem~\ref{theorem:rdm-coarse}, we get
    \begin{align}
        \rho_{AB}&=\sum_{i}\lam_i\left(\lim_{n\to\infty}\frac{1}{n}\sum_{k=1}^n\mcl{C}_{\povm{P}_A,\gm_A}^k(\sigma_A^i)\right)\ot\sigma_B^i\\
        &=\sum_{i}\lam_i\;\Delta_{\povm{P}_A,\gm_A}(\sigma_A^i)\ot\sigma_B^i\;.
    \end{align}
    Thus, we have~(\ref{theorem:locally-micro-states:coarse-graining})$\implies$(\ref{theorem:locally-micro-states:resource-destroying}). Conversely,~(\ref{theorem:locally-micro-states:resource-destroying})$\implies$(\ref{theorem:locally-micro-states:coarse-graining}) follows from the fact that $\Delta_{\povm{P}_A,\gm_A}$ is a resource destroying map.
\end{proof}

\begin{remark}
    We note that Theorem~\ref{theorem:locally-micro-states} is derived under the assumption of a prior state in tensor-product form, i.e., $\gm_{AB}\coloneq\gm_A\ot\gm_B$. This simplification enables a clear characterization of locally macroscopic states. In general, for a correlated prior $\gm_{AB}$, it becomes $\mcl{C}_{\povm{P}_A,\gm_{AB}}\neq\mcl{C}_{\povm{P}_A,\gm_A}\ot\mcl{I}_B$. Therefore, whether Theorem~\ref{theorem:locally-micro-states} holds for a prior that includes general non-tensor-product states remains an open question.
\end{remark}

In addition, it is possible to define \emph{locally macroscopic operations}. The set of locally macroscopic states with respect to POVM $\povm{P}_A$ and prior state $\gm_{AB}=\gm_A\ot\gm_B$ is denoted $\mf{L}(\povm{P}_A,\gm_{AB})$. 

\begin{itemize}%[label=(\alph*)]
    \item \textit{Local Microscopicity Non-Generating Operations (LMNO)}: a CPTP map $\mcl{E}_{AB}$ is LMNO whenever
    \begin{align}
        \mcl{E}_{AB}(\rho_{AB})\in\mf{L}(\povm{P}_A,\gm_{AB})\;,    
    \end{align}
    for all $\rho_{AB}\in\mf{L}(\povm{P}_A,\gm_{AB})$.
    \item \textit{Local Resource-Destroying Covariant Operations (LRCO)}: a CPTP map $\mcl{E}_{AB}$ is LRCO whenever
    \begin{align}
        \mcl{E}_{AB}\circ(\Delta_{\povm{P}_A,\gm_A}\ot\mcl{I}_B)=(\Delta_{\povm{P}_A,\gm_A}\ot\mcl{I}_B)\circ\mcl{E}_{AB}\;.
    \end{align}
    \item \textit{Local Coarse-Graining Covariant Operations (LCCO)}: a CPTP map $\mcl{E}_{AB}$ is LCCO whenever
    \begin{align}
        \mcl{E}_{AB}\circ(\mcl{C}_{\povm{P}_A,\gm_A}\ot\mcl{I}_B)=(\mcl{C}_{\povm{P}_A,\gm_A}\ot\mcl{I}_B)\circ\mcl{E}_{AB}\;.
    \end{align}
\end{itemize}

Similar to the resource theory of microscopicity (see Proposition~\ref{proposition:CCOMNO}), the following hierarchy holds:
\begin{align}
    \mr{LCCC}\subseteq\mr{LRCO}\subseteq\mr{LMNO}\;.
\end{align}

\subsection{Observational discord}
\label{subsection:locally-micro-correlations}

Based on the concept of locally macroscopic states, we introduce a new measure of quantum correlation that depends on an observer's measurement capabilities. While standard measures, such as quantum (one-way) deficit~\cite{oppenheim2002thermodynamical,horodecki2003local,devetak2005distillation,horodecki2005local} and discord~\cite{ollivier2001quantum,modi2012classical,ming2018quantum}, quantify an intrinsic property of a quantum state by optimizing over all possible measurements, our approach addresses a different operational question: ``When is a local observer with fixed -- and possibly limited -- measurement capabilities able to access all the correlations present in a bipartite state?'' This question is physically relevant in many scenarios where optimization is not possible. The measure we introduce, \textit{observational discord}, quantifies the amount of total correlation, measured by quantum mutual information, that is inaccessible to a specific local observer characterized by the pair $(\povm{P}_A,\rho_A)$.

In what follows, $\mcl{M}_{\povm{P}_A}$ is the quantum-classical channel corresponding to the POVM $\povm{P}_A$ on the Hilbert space $\sH_A$, as in Eq.~\eqref{eq:qc-channel}.

\begin{definition}[Observational discord]
\label{definition:ob-discord}
    Let $\rho_{AB}\in\mcl{S}(\sH_A\ot\sH_B)$. Then, the \emph{observational discord} is defined as
    \begin{align}
        D_{\povm{P}_A}(\bar{A};B)_{\rho_{AB}}&\coloneq D(\rho_{AB}\|\rho_A\ot\rho_B)-D\Big((\mcl{M}_{\povm{P}_A}\ot\mcl{I}_B)(\rho_{AB})\|(\mcl{M}_{\povm{P}_A}\ot\mcl{I}_B)(\rho_A\ot\rho_B)\Big)\\
        &=I(A;B)_{\rho_{AB}}-I(X;B)_{\omega_{XB}}\;,
    \end{align}
    with the overbar denoting the measured system, $\omega_{XB}\coloneq(\mcl{M}_{\povm{P}_A}\ot\mcl{I}_B)(\rho_{AB})$ and $I(A;B)_{\rho_{AB}}\coloneq D(\rho_{AB}\|\rho_A\ot\rho_B)$ is the quantum mutual information.
\end{definition}

Note that in the above definition, the prior is taken to be the product of the marginals of $\rho_{AB}$, instead of an arbitrary tensor product state, as in the previous subsection.

\begin{remark}[Quantum discord]
    While the observational discord is defined for a fixed POVM acting on $A$, optimizing over all POVMs we obtain the quantum discord~\cite{ollivier2001quantum,modi2012classical,ming2018quantum,seshadreesan2015Renyi}
    \begin{align}
        \inf_{\povm{P}_A:\;\mr{POVM}}D_{\povm{P}_A}(\bar{A};B)_{\rho_{AB}}=I(A;B)_{\rho_{AB}}-\sup_{\povm{P}_A:\;\mr{POVM}}I(X;B)_{\omega_{XB}}\;.
    \end{align}
    The minimum is achieved with a rank-one POVM~\cite{modi2012classical,seshadreesan2015Renyi,seshadreesan2015fidelity}. 
\end{remark}

In the precise sense that discord is optimized over all POVMs, discord is observer-independent. However, our observational discord depends on the choice of a particular observer's macroscopic viewpoint.

As a consequence of Theorem~\ref{theorem:locally-micro-states}, we can derive the necessary and sufficient conditions for vanishing observational discord.

\begin{corollary}
\label{corollary:locally-micro-correlation}
    Let $\rho_{AB}\in\mcl{S}(\sH_A\ot\sH_B)$ and let $\povm{\Pi}_{\povm{P}_A,\rho_A}\equiv\povm{\Pi}=\{\Pi_y\}_{y\in\set{Y}}$ be the MPPP with respect to POVM $\povm{P}_A$ and prior state $\rho_A$. Then, the following conditions are equivalent:
    \begin{enumerate}
        \item $D_{\povm{P}_A}(\bar{A};B)_{\rho_{AB}}=0$;
        \item $(\mcl{C}_{\povm{P}_A,\rho_A}\ot\mcl{I}_B)(\rho_{AB})=(\mcl{R}_{\mcl{M}_{{\povm{P}}_A}\ot\mcl{I}_B,\;\rho_A\ot\rho_B})\circ(\mcl{M}_{\povm{P}_A}\ot\mcl{I}_B)(\rho_{AB})=\rho_{AB}$;
        \item $(\Delta_{\povm{P}_A,\rho_A}\ot\mcl{I}_B)(\rho_{AB})=\rho_{AB}$;
        \item there exist $c_y^i\geq0$, $\lam_i\geq0$ satisfying $\sum_i\lam_i=1$, $\sigma_A^i\in\mcl{S}(\sH_A)$, and $\sigma_B^i\in\mcl{S}(\sH_B)$ such that
        \begin{align}
            \rho_{AB}&=\sum_{i}\lam_i\;\sigma_A^i\ot\sigma_B^i\\
            &=\sum_{i}\lam_i\;\Delta_{\povm{P}_A,\rho_A}(\sigma_A^i)\ot\sigma_B^i\\
            &=\sum_{i}\lam_i\left(\sum_{y\in\set{Y}}c_y^i\Pi_y\rho_A\right)\ot\sigma_B^i\;.
        \end{align}
    \end{enumerate}
\end{corollary}

Corollary~\ref{corollary:locally-micro-correlation} implies that the correlation, as quantified by Definition~\ref{definition:ob-discord}, vanishes when one part of the bipartite system is coarse-grained. This contrasts with the observer-independent nature of quantum (one-way) deficit~\cite{oppenheim2002thermodynamical,horodecki2003local,devetak2005distillation,horodecki2005local} and discord~\cite{ollivier2001quantum,modi2012classical,ming2018quantum}.  Our measure (Definition~\ref{definition:ob-discord}) thus highlights a more realistic aspect of quantum correlations: their perceived ``quantumness'' is not absolute but depends on the capabilities and perspective of the observer.

In the case of quantum discord~\cite{modi2012classical,ming2018quantum}, the necessary and sufficient condition for being
\begin{align}
    \inf_{\povm{P}_A:\;\mr{POVM}}D_{\povm{P}_A}(\bar{A};B)_{\rho_{AB}}=0
\end{align}
is that $\rho_{AB}$ is a classical-quantum state, i.e., $\rho_{AB}=\sum_i\lam_i\ketbra{i}_A\ot\sigma_B^i$ for some local orthonormal basis $\{\ket{i}_A\}_{i}$, $\lam_i\geq0$, $\sum_i\lam_i=1$, and $\sigma_B\in\mcl{S}(\sH_B)$. Let $\mf{CQ}$ be the set of classical-quantum states. Then,
\begin{align}
    \inf_{\sigma_{AB}\in\mf{CQ}}D(\bigcdot\|\sigma_{AB})
\end{align}
is called \emph{relative entropy of discord}~\cite{streltsov2017colloquium,modi2012classical,ming2018quantum}. 

Similarly, we can define the \emph{relative entropy of local microscopicity} as follows:
\begin{align}
    \inf_{\sigma_{AB}\in\mf{L}(\povm{P}_A,\gm_{AB})}D(\bigcdot\|\sigma_{AB})\;,
\end{align}
where $\mf{L}(\povm{P}_A,\gm_{AB})$ is the set of locally macroscopic states with respect to POVM $\povm{P}_A$ and the prior state $\gm_{AB}$. Note that while $\mf{L}(\povm{P}_A,\gm_{AB})$ is a convex set, $\mf{CQ}$ is not. We leave it open the problem of finding a more quantitative relationship between the relative entropy of local microscopicity and our Definition~\ref{definition:ob-discord}.

\section{Conclusion}

In this paper, we developed a general framework for understanding macroscopic states in quantum systems, grounded in the interplay between general POVMs and quantum priors. We defined macroscopic states operationally as fixed points of a coarse-graining map, i.e., the composition of a measurement and its Petz recovery map, and provided several equivalent characterizations. These include conditions based on quantum statistical sufficiency, the structure of maximal projective post-processings, and the fixed points of a corresponding resource-destroying map. Together, these results offer a unified perspective on macroscopic irreversibility and open the way for a generalized resource-theoretic treatment of coarse-graining and retrodiction in quantum theory. These developments provide a new perspective on the nature of quantum correlations within realistic and constrained settings.

Future work includes extending our framework to the case of incompatible POVMs, where the structure of observational entropy and macroscopic states becomes more subtle. In this context, a rigorous definition of generalized observational entropy and a corresponding law of irreducibility (i.e., monotonicity under macroscopic operations) remain open questions. Moreover, a key challenge is to formulate an appropriate notion of the tensor product for macroscopic states, which is essential for analyzing composite systems and catalytic transformations.

An intriguing direction is the connection between the relative entropy of microscopicity and correlated catalytic transformations. In the resource theory of coherence, it is known that 
\begin{equation}
D\left(\rho\Big\|\mathcal{P}_{\{\ketbra{i}\}_i}(\rho)\right) \geq D\left(\sigma\Big\|\mathcal{P}_{\{\ketbra{i}\}_i}(\sigma)\right)
\end{equation}
if and only if a correlated catalytic state conversion from $\rho$ to $\sigma$ is possible under dephasing-covariant incoherent operations (DIO)~\cite{datta2023there,takagi2022correlation}. Similar to results in the resource theory of coherence and building on the generalized quantum Stein's lemma~\cite{hayashi2024generalized,lami2024solution}, we conjecture that converting a state from $\rho$ to $\sigma$ via free macroscopic operations with a correlated catalyst might be possible if the following condition holds: 
\begin{equation}
    D\left(\rho\Big\|\Delta_{\povm{P},\gamma}(\rho)\right) \geq D\left(\sigma\Big\|\Delta_{\povm{P},\gamma}(\sigma)\right)\;.
\end{equation}
Proving this conjecture remains an open problem.

% By leveraging the generalized quantum Stein’s lemma~\cite{hayashi2024generalized,lami2024solution}, we expect a similar equivalence to hold in our setting:
% \begin{equation}
% D\left(\rho\Big\|\Delta_{\povm{P},\gamma}(\rho)\right) \geq D\left(\sigma\Big\|\Delta_{\povm{P},\gamma}(\sigma)\right)
% \end{equation}
% if and only if there exists a correlated catalytic transformation from $\rho$ to $\sigma$ under free macroscopic operations~\cite{datta2023there,takagi2022correlation}.

Furthermore, we can discuss the implications of this framework for quantum memory usefulness~\cite{takagi2024when}, particularly by exploring how the inferential reference frame (i.e., the prior used to evaluate correlations) determines the operational value of entanglement in communication protocols where encoding operations are restricted by coarse-graining~\cite{korzekwa2022encoding,hayashi2022dense}. We can also consider the connection between irreversible symmetries imposed on quantum systems by the coarse-graining map~\cite{okada2024noninvertible} and the limits of possible measurements, particularly in light of foundational results like the Wigner-Araki-Yanase theorem~\cite{loveridge2020relational}.

Finally, a promising line of investigation is to use this framework to analyze entropy production in coarse-grained quantum systems, possibly unifying recent approaches to fluctuating entropy production and thermodynamic consistency~\cite{degunther2024fluctuating,bai2024fully,dieball2025perspective}.

\section*{Acknowledgments}
The authors are grateful to Niklas Galke, Masahito Hayashi, Anna Jencova, Shintaro Minagawa, Joseph Schindler, and Ryuji Takagi for suggestions that improved the manuscript. T.~N. acknowledges the ``Nagoya University Interdisciplinary Frontier Fellowship'' supported by Nagoya University and JST, the establishment of university fellowships towards the creation of science technology innovation, Grant Number JPMJFS2120 and ``THERS Make New Standards Program for the Next Generation Researchers'' supported by JST SPRING, Grant Number JPMJSP2125. K.~K. acknowledges support from JSPS Grant-in-Aid for Early-Career Scientists, No. 22K13972; and from the MEXT-JSPS Grant-in-Aid for Transformative Research Areas (B) No. 24H00829. 
K.~K, E.~W. and F.~B. acknowledge support from MEXT Quantum Leap Flagship Program (MEXT QLEAP) Grant No. JPMXS0120319794.
F.~B. also acknowledges  support from MEXT-JSPS  Grant-in-Aid for Transformative Research Areas  (A) ``Extreme Universe,'' No.~21H05183, and  from JSPS  KAKENHI Grant No.~23K03230.

\printbibliography %Prints bibliography

\end{document}